\documentclass[12pt]{article}
\usepackage{amsmath}
\usepackage{graphicx,psfrag,epsf}
\usepackage{enumerate}
\usepackage{url}
\usepackage{amssymb}
\usepackage{amsfonts}
\usepackage[nohead]{geometry}
\usepackage[singlespacing]{setspace}
\usepackage[bottom]{footmisc}
\usepackage{indentfirst}
\usepackage{endnotes}
\usepackage{rotating}
\usepackage{amsthm}
\usepackage{booktabs}
\usepackage{algorithm}
\usepackage{array}
\usepackage{multirow}
\usepackage{algpseudocode}
\usepackage{enumerate}
\usepackage{bbm}
\usepackage{dsfont}
\usepackage{verbatim}
\usepackage{subfig}
\usepackage{xr}
\usepackage{tikz}
\usepackage{float}
\RequirePackage[OT1]{fontenc} \RequirePackage{amsthm,amsmath}
\RequirePackage[numbers]{natbib}
\newcommand{\field}[1]{\mathbb{#1}}

\newcommand{\X}{\field{X}}
\newcommand{\E}{\field{E}}
\def\zhang{\textcolor{red}}
\def\chakraborty{\textcolor{blue}}

\theoremstyle{example} \theoremstyle{remark} \theoremstyle{lemma}
\theoremstyle{definition} \theoremstyle{corol}
\theoremstyle{proposition} \theoremstyle{condition}
\theoremstyle{assumption}
\newtheorem{assumption}{\n{Assumption}}[section]
\newtheorem{example}{\n{Example}}[section]
\newtheorem{remark}{\n{Remark}}[section]
\newtheorem{lemma}{\n{Lemma}}[section]
\newtheorem{definition}{\n{Definition}}[section]

\newtheorem{proposition}{\n{Proposition}}[section]

\newcommand{\Rmnum}[1]{\expandafter\romannumeral #1}
\newcommand*{\I}{\imath}

\font\n=cmcsc12

\def\X{{\mathbf{X}}}

\newcommand\independent{\protect\mathpalette{\protect\independenT}{\perp}}
\def\independenT#1#2{\mathrel{\rlap{$#1#2$}\mkern2mu{#1#2}}}

  \textwidth = 420pt
\renewcommand{\baselinestretch}{1}

\makeatletter
\renewcommand\@biblabel[1]{}
\renewenvironment{thebibliography}[1]
     {\section*{\refname}%
      \@mkboth{\MakeUppercase\refname}{\MakeUppercase\refname}%
      \list{}%
           {\leftmargin0pt
            \@openbib@code
            \usecounter{enumiv}}%
      \sloppy
      \clubpenalty4000
      \@clubpenalty \clubpenalty
      \widowpenalty4000%
      \sfcode`\.\@m}
     {\def\@noitemerr
       {\@latex@warning{Empty `thebibliography' environment}}%
      \endlist}
\makeatother

\begin{document}

\def\spacingset#1{\renewcommand{\baselinestretch}%
{#1}\small\normalsize} \spacingset{1}


  \title{\bf Distance Metrics for Measuring Joint Dependence with Application to Causal Inference}
  \author{Shubhadeep Chakraborty \hspace{.2cm}\\
    Department of Statistics, Texas A\&M University\\
    and \\
    Xianyang Zhang \\
    Department of Statistics, Texas A\&M University}
  \maketitle

\bigskip
\begin{abstract}
Many statistical applications require the
quantification of joint dependence among more than two random
vectors. In this work, we generalize the notion of distance
covariance to quantify joint dependence
among $d \geq 2$ random vectors. We introduce the high order
distance covariance to measure the so-called Lancaster interaction dependence. The joint distance covariance is then
defined as a linear combination of pairwise distance covariances and
their higher order counterparts which together completely
characterize mutual independence. We further introduce some related
concepts including the distance cumulant, distance characteristic
function, and rank-based distance covariance. Empirical estimators are constructed based on certain Euclidean
distances between sample elements. We study the large sample
properties of the estimators and propose a bootstrap procedure to
approximate their sampling distributions. The asymptotic validity of
the bootstrap procedure is justified under both the null and
alternative hypotheses. The new metrics are employed to perform
model selection in causal inference, which is based on the joint
independence testing of the residuals from the fitted structural
equation models. The effectiveness of the method is
illustrated via both simulated and real datasets.
\end{abstract}

\noindent%
\vfill

\newpage
\spacingset{1.45} 

\section{Introduction}
Measuring and testing dependence is of central importance in
statistics, which has found applications in a wide variety of areas
including independent component analysis, gene selection, graphical
modeling and causal inference. Statistical tests of
independence can be associated with widely many dependence measures.
Two of the most classical measures of association between two
ordinal random variables are Spearman's rho and Kendall's tau.
However, tests for (pairwise) independence using these two classical
measures of association are not consistent, and only have power for
alternatives with monotonic association. Contingency table-based
methods, and in particular the power-divergence family of test
statistics (Read and Cressie, 1988), are the best known general
purpose tests of independence, but are limited to relatively low
dimensions, since they require a partitioning of the space in which
each random variable resides. Another classical measure of
dependence between two random vectors is the mutual information
(Cover and Thomas, 1991), which can be interpreted as the
Kullback-Leibler divergence between the joint density and the
product of the marginal densities. The idea originally dates back to
the 1950's, in groundbreaking works by Shannon and Weaver (1949),
Mcgill (1954) and Fano (1961). Mutual information completely
characterizes independence and generalizes to more than two random
vectors. However, test based on mutual information involves
distributional assumptions for the random vectors and hence is not
robust to model misspecification.

In the past fifteen years, kernel-based methods have received
considerable attention in both the statistics and machine learning literature. For instance, Bach and Jordan (2002) derived a regularized correlation operator from the covariance and
cross-covariance operators and used its largest singular value to conduct independence test. Gretton et al. (2005; 2007) introduced a kernel-based independence measure, namely the Hilbert-Schmidt Independence Criterion (HSIC), to test for independence of two random vectors. This idea was recently extended by Sejdinovic et al. (2013) and Pfister et al. (2018) to quantify the joint independence among more than two random vectors.


Along with a different direction, Sz\'{e}kely et al. (2007), in their
seminal paper, introduced the notion of distance covariance (dCov) and
distance correlation as a measure of dependence between two random
vectors of arbitrary dimensions. Given the theoretical appeal of the
population quantity and the striking simplicity of the sample
version, the idea has been widely extended and analyzed in various
ways in Sz\'{e}kely and Rizzo (2012; 2014), Lyons (2013), Sejdinovic
et al. (2013), Dueck et al. (2014), Bergsma et al. (2014), Wang et
al. (2015), and Huo and Sz\'{e}kely (2016), to mention only a few.
The dCov between two random vectors $X \in {\mathbb
R}^p$ and $Y \in {\mathbb R}^q$ with finite first moments is defined
as the positive square root of 
\begin{eqnarray*}
dCov^2(X,Y)=\frac{1}{c_{p}c_{q}}\int_{{\mathbb
R}^{p+q}}\frac{|f_{X,Y}(t,s)-f_X(t)f_Y(s)|^2}{|t|_p^{1+p}|s|_q^{1+q}}dtds,
\end{eqnarray*}
where $f_X$, $f_Y$ and $f_{X,Y}$ are the individual and joint
characteristic functions of $X$ and $Y$ respectively, $|\cdot|_p$
is the Euclidean norm of $\mathbb{R}^p$, $c_{p}=\pi^{(1+p)/2}/\,\Gamma((1+p)/2)$
is a constant with $\Gamma(\cdot)$ being the complete gamma function. An
important feature of dCov is that it fully
characterizes independence because $dCov(X,Y)=0$ if and
only if $X$ and $Y$ are independent.

Many statistical applications require the quantification of joint
dependence among $d\geq 2$ random variables (or vectors). Examples
include model diagnostic checking for directed acyclic graph (DAG) where
inferring pairwise independence is not enough in this case (see more
details in Section \ref{sec:dag}), and independent component
analysis which is a means for finding a suitable
representation of multivariate data such that the components of the transformed data are mutually independent.
In this paper, we shall introduce new metrics which generalize the
notion of dCov to quantify joint dependence of $d
\geq 2$ random vectors. We first introduce the notion of high order
dCov to measure the so-called Lancaster interaction
dependence (Lancaster, 1969). We generalize the notion of Brownian covariance (Sz\'{e}kely et al., 2009) and show that it coincides with the high order distance covariance. We then define the joint dCov (Jdcov) as a linear combination of pairwise dCov and their high order counterparts. The proposed metric
provides a natural decomposition of joint dependence into the sum
of lower order and high order effects, where the relative importance
of the lower order effect terms and the high order effect terms is
determined by a user-chosen number. In the population case, Jdcov is
equal to zero if and only if the $d$ random vectors are mutually
independent, and thus completely characterizes joint independence.
It is also worth mentioning that the proposed metrics are invariant to
permutation of the variables and they inherit some nice properties
of dCov, see Section \ref{sec:jdov}.

Following the idea of Streitberg (1990), we introduce the concept of
distance cumulant and distance characteristic function, which leads
us to an equivalent characterization of independence of the $d$
random vectors. Furthermore, we establish a scale invariant
version of Jdcov and discuss the concept of rank-based distance
measures, which can be viewed as the counterparts of Spearman's rho
to dCov and JdCov.

JdCov and its scale-invariant versions can be conveniently estimated
in finite sample using V-statistics or their bias-corrected
versions. We study the asymptotic properties of the estimators, and
introduce a bootstrap procedure to approximate their sampling
distributions. The asymptotic validity of the bootstrap procedure is
justified under both the null and alternative hypotheses. The new
metrics are employed to perform model selection in a causal
inference problem, which is based on the joint independence testing
of the residuals from the fitted structural equation models. We
compare our tests with the bootstrap version of the $d$-variate HSIC
(dHSIC) test recently introduced in Pfister et al. (2018) and the
mutual independence test proposed by Matteson and Tsay (2017).
Finally we remark that although we focus on Euclidean space valued
random variables, our results can be readily extended to general
metric spaces in view of the results in Lyons (2013).

The rest of the paper is organized as follows. Section \ref{sec:high} introduces the high order distance covariance and studies its basic properties.
Section \ref{sec:jdov} describes the JdCov to quantity joint dependence. Sections \ref{sec:cum}-\ref{sec:rank}
further introduce some related concepts including the distance cumulant, distance characteristic
function, and rank-based distance covariance. We study the estimation of the distance metrics in Section \ref{sec:est} and present a joint independence test
based on the proposed metrics in Section \ref{sec:infer}.
Section \ref{sec:num} is devoted to numerical studies. The new metrics are employed to perform model selection in causal inference in
Section \ref{sec:dag}. Section \ref{sec:comp} discusses the efficient computation of distance metrics and future research directions.
The technical details are gathered in the supplementary material.

\emph{Notations}. Consider $d\geq 2$ random vectors
$\mathcal{X}=\{X_1,\dots,X_d\}$, where $X_i\in\mathbb{R}^{p_i}$. Set $p_0=\sum^{d}_{i=1}p_i$. Let $\{X_1',\dots,X_d'\}$ be an
independent copy of $\mathcal{X}$. Denote by $\I=\sqrt{-1}$ the
imaginary unit. Let $|\cdot|_p$ be the Euclidean norm of
$\mathbb{R}^p$ with the subscript omitted later without ambiguity.
For $a,b\in\mathbb{R}^p$, let $\langle a,b\rangle=a^{\top}b$. For a
complex number $a$, denote by $\bar{a}$ its conjugate. Let $f_i$ be
the characteristic function of $X_i$, i.e., $f_i(t)=\E[e^{\I \langle
t,X_i\rangle}]$ with $t\in\mathbb{R}^{p_i}$. Define
$w_{p}(t)=(c_{p}|t|^{1+p}_p)^{-1}$ with
$c_{p}=\pi^{(1+p)/2}/\,\Gamma((1+p)/2).$
Write $dw=(c_{p_1}c_{p_2}\dots
c_{p_d}|t_1|_{p_1}^{1+p_1}\cdots|t_d|_{p_d}^{1+p_d}
)^{-1}dt_1\cdots dt_d.$ Let $I^d_k$ be the collection of $k$-tuples
of indices from $\{1,2,\dots,d\}$ such that each index occurs
exactly once. Denote by $\lfloor a \rfloor $ the integer part of $a\in\mathbb{R}$.
Write $X \independent Y$ if $X$ is independent of $Y.$

\section{Measuring joint dependence}\label{sec:joint}
\subsection{High order distance covariance}\label{sec:high}
We briefly review the concept of Lancaster interactions first introduced by Lancaster (1969).
The Lancaster interaction measure associated with a multidimensional probability distribution of $d$ random variables $\{X_1,\dots,X_d\}$ with the
joint distribution $F=F_{1,2,\dots,d}$, is a signed measure $\Delta F$ given by
\begin{equation}\label{eq-land}
\Delta F \, = \, (F_1^* - F_1) (F_2^* - F_2) \cdots (F_d^* - F_d)\,,
\end{equation}
where after expansion, a product of the form $F_i^* F_j^*\cdots F_k^*$ denotes the corresponding joint distribution function $F_{i,j,\dots,k}$ of $\{X_i,X_j,\dots,X_k\}$. For example for $d=4$, the term $F_1^* F_2^* F_3 F_4$ stands for $F_{12} F_3 F_4$, $F_1^* F_2 F_3 F_4$ stands for $F_1 F_2 F_3 F_4$, etc.
In particular for $d = 3$, (\ref{eq-land}) simplifies to
\begin{equation}\label{eq-lan3}
\Delta F \, = \, F_{123} - F_1 F_{23} - F_2 F_{13} - F_3 F_{12} + 2 F_1 F_2 F_3\, .
\end{equation}
In light of the Lancaster interaction measure, we introduce the concept of $d$th order dCov as follows.

\begin{definition}
The $d$th order dCov is defined as the
positive square root of
\begin{align}
\label{eqn1}
dCov^2(X_1,\dots,X_d)=&\int_{\mathbb{R}^{p_0}}\left|\E\left[\prod_{i=1}^{d}(f_i(t_i)-e^{\I
\langle t_i,X_i\rangle})\right]\right|^2dw,
\end{align}
When $d=2$, it reduces to the dCov in Sz\'{e}kely et al. (2007).
\end{definition}

The term $\E[\prod_{i=1}^{d}(f_i(t_i)-e^{\I \langle
t_i,X_i\rangle})]$ in the definition of dCov is a counterpart of the Lancaster interaction measure in (\ref{eq-land}) with the joint distribution functions replaced by
the joint characteristic functions. When $d=3$, $dCov^2(X_1,X_2,X_3)>0$ rules out the possibility
of any factorization of the joint distribution.
To see this, we note that $X_1\independent (X_2,X_3)$, $X_2\independent (X_1,X_3)$
or $X_3\independent (X_1,X_2)$ all lead to
$dCov^2(X_1,X_2,X_3)=0$. On the other hand,
$dCov^2(X_1,X_2,X_3)=0$ implies that
\begin{align*}
&f_{123}(t_1,t_2,t_3)-f_1(t_1)f_2(t_2)f_3(t_3)
\\=&f_1(t_1)f_{23}(t_2,t_3)+f_2(t_2)f_{13}(t_1,t_3)+f_3(t_3)f_{12}(t_1,t_2)-3f_1(t_1)f_2(t_2)f_3(t_3)
\end{align*}
for $t_i\in\mathbb{R}^{p_i}$ almost everywhere. In this case, the
``higher order effect'' i.e.,
$f_{123}(t_1,t_2,t_3)-f_1(t_1)f_2(t_2)f_3(t_3)$ can be represented
by the ``lower order/pairwise effects''
$f_{ij}(t_i,t_j)-f_i(t_i)f_j(t_j)$ for $1\leq i\neq j\leq 3$. However, this does not necessarily imply that $X_1, X_2$ and $X_3$ are jointly independent. In other words when $d=3$ (or more generally when $d \geq 3$), joint independence of $X_1, X_2$ and $X_3$ is not a necessary condition for dCov to be zero. To address this issue, we shall introduce a new distance metric to quantify any forms of dependence among $\mathcal{X}$ in Section \ref{sec:jdov}.

In the following, we present some basic properties of high order dCov. Define the bivariate
function
$U_i(x,x')=\E|x-X'_i|+\E|X_i-x'|-|x-x'|-\E|X_i-X'_i|$
for $x,x'\in\mathbb{R}^{p_i}$ with $1\leq i\leq d$. Our definition of dCov is partly motivated by the following lemma.
\begin{lemma}\label{lem1}
For $1\leq i\leq d,$
$$U_i(x,x')=\int_{\mathbb{R}^{p_i}}\left\{(f_i(t)-e^{\I
\langle t, x\rangle})(f_i(-t)-e^{-\I \langle
t,x'\rangle})\right\}w_{p_i}(t)dt.$$
\end{lemma}
By Lemma \ref{lem1} and Fubini's theorem, the $d$th order (squared) dCov admits the following equivalent representation,
\begin{equation}
\label{eqn2}
\begin{split}
dCov^2(X_1,\dots,X_d)
=&\int_{\mathbb{R}^{p_0}}\left|\E\left[\prod_{i=1}^{d}(f_i(t_i)-e^{\I \langle t_i,X_i\rangle})\right]\right|^2dw\\
=&\int_{\mathbb{R}^{p_0}}\E\left[\prod_{i=1}^{d}(f_i(t_i)-e^{\I
\langle t_i,
X_i\rangle})\right]\E\left[\prod_{i=1}^{d}\overline{(f_i(t_i)-e^{\I
\langle t_i, X_i'\rangle})}\right]dw \\
=& \,\E\left[\prod^{d}_{i=1}
U_i(X_i,X'_i)\right].
\end{split}
\end{equation}
This suggests that similar to dCov, its high order counterpart has an
expression based on the moments of $U_i$s, which results in very
simple and applicable empirical formulas, see more details in Section \ref{sec:est}.

\begin{remark}\label{rem1}
{\rm
From the definition of dCov in Sz\'{e}kely et al. (2007), it might appear that its most natural generalization to the case of $d = 3$ would be to define a measure in the following way
$$\frac{1}{c_{p}c_{q}c_{r}}\int_{{\mathbb
R}^{p+q+r}}\frac{|f_{X,Y,Z}(t,s,u)-f_X(t)f_Y(s)f_Z(u)|^2}{|t|_p^{1+p}|s|_q^{1+q}|u|_r^{1+r}}\,dtdsdu \, ,$$
where $X\in\mathbb{R}^p$, $Y\in\mathbb{R}^q$ and $Z \in \mathbb{R}^r$. Assuming that the integral above exists, one can easily verify that such a measure completely
characterizes joint independence among $X, Y$ and $Z$. 
However, it does not admit a nice equivalent representation as in (\ref{eqn2}) (unless one considers a different weighting function). We exploit this equivalent representation of the $d$th order dCov to propose a V-statistic type estimator of the population quantity (see Section 3) which is much simpler to compute rather than evaluating an integral as in the original definition in (\ref{eqn1}).
}
\end{remark}

\begin{remark}\label{rem2}
{\rm
Sz\'{e}kely et al. (2009) introduced the notion of covariance with respect to a stochastic process. Theorem $8$ in Sz\'{e}kely et al. (2009) shows that population distance covariance coincides with the covariance with respect to Brownian motion (or the so-called \textit{Brownian covariance}). The Brownian covariance of two random variables $X\in \mathbb{R}^p$ and $Y \in \mathbb{R}^q$ with $\mathbb{E}(|X|^2+|Y|^2)<\infty$ is defined as the positive square root of $$\mathcal{W}^2(X,Y)=Cov^2_W(X,Y)=\mathbb{E}[X_W X^{'}_W Y_{W^{'}} Y^{'}_{W^{'}}]\, ,$$ where $W$ and $W^{'}$ are independent Brownian motions with zero mean and covariance function $C(t,s)=|s| + |t| - |s-t|$ on $\mathbb{R}^p$ and $\mathbb{R}^q$ respectively, and $$X_W = W(X) - \mathbb{E}[\,W(X)|W\,]\, .$$ Conditional on $W$ (or $W^{'}$), $X^{'}_W$ (or $Y^{'}_{W^{'}}$) is an i.i.d. copy of $X_W$ (or $Y_{W^{'}}$). Then following Theorem $8$ in Sz\'{e}kely et al. (2009) and Definition $2.1$, we have $dCov^2(X,Y) = \mathcal{W}^2(X,Y) \,.$}
\end{remark}

Now for $d \geq 2$ random variables $\{X_1, X_2, \dots , X_d\}$ where $X_i \in \mathbb{R}^{p_i}, 1\leq i \leq d$, we can generalize the notion of Brownian covariance as the positive square root of $$\mathcal{W}^2(X_1, \dots , X_d) = \mathbb{E}\left[ \prod_{i=1}^d X_{i_{W_i}} X^{'}_{i_{W_i}} \right]\, ,$$ where $W_i$'s are independent Brownian motions on $\mathbb{R}^{p_i}$, $1\leq i \leq d$. Property $(2)$ in Proposition $2.1$ below establishes the connection between the higher order distance covariances and the generalized notion of Brownian covariance.

Similar to $dCov$, our
definition of high order $dCov$ possesses the following
important properties.
\begin{proposition}\label{prop0}
We have the following properties regarding $dCov(X_1, X_2, \dots , X_d)$:
\begin{enumerate}[(1)]
\item For any $a_i\in\mathbb{R}^{p_i}$, $c_i\in\mathbb{R}$, and orthogonal transformations $A_i\in\mathbb{R}^{p_i\times p_i}$, $dCov^2(a_1+c_1A_1X_1,\dots,a_d+c_dA_dX_d)=\prod_{i=1}^{d}|c_i|\,
dCov^2(X_1,\dots,X_d).$ Moreover, dCov is invariant to any permutation of $\{X_1,X_2,\\ \dots,X_d\}$.
\item Under Assumption \ref{ass1} (see Section \ref{sec:est}), the $d$th order $dCov$ exists and $$\mathcal{W}^2(X_1, \dots , X_d) = dCov^2(X_1, \dots , X_d)\,.$$
\end{enumerate}
\end{proposition}
Property (1) shows that
dCov is invariant to translation, orthogonal transformation and permutation on
$X_i$s. In property (2), the existence of the $d$th order $dCov$ follows from (\ref{eqn2}) and application of Fubini's Theorem and H\"{o}lder's
inequality. The equality with Brownian covariance readily follows from the proof of Theorem $7$ in Sz\'{e}kely et al. (2009).

Theorem 7 in Sz\'{e}kely et al. (2007) shows the relationship between distance correlation and the correlation coefficient for bivariate normal distributions. We extend that result in case of multivariate normal random variables with zero mean, unit variance and pairwise correlation $\rho$. Proposition \ref{prop:normality} below establishes a relationship between the correlation coefficient and higher order distance covariances for multivariate normal random variables.

\begin{proposition}\label{prop:normality}
Suppose $(X_1, X_2, \dots , X_d)\sim N(0,\Sigma)$, where
$\Sigma=(\sigma_{i,j})_{i,j=1}^{d}$ with $\sigma_{ii}=1$ for $1\leq
i\leq d$ and $\sigma_{ij}=\rho$ for $1\leq
i \neq j \leq d$. When $d=2k-1$ or $d=2k$, $dCov^2(X_1,\dots ,X_d)= O(|\rho|^{2k})$ for $k \geq 2$.
\end{proposition}

Proposition \ref{prop01} in the supplementary materials shows some additional properties of the $d$th order $dCov$.

\subsection{Joint distance covariance}\label{sec:jdov}
In this subsection, we introduce a new joint dependence measure
called the joint dCov (Jdcov), which is designed to
capture all types of interaction dependence among the $d$ random
vectors. To achieve this goal, we define JdCov as the linear
combination of all $k$th order dCov for $1\leq k\leq
d$.
\begin{definition}
The JdCov among $\{X_1,\dots,X_d\}$ is given by
\begin{equation}\label{def1}
\begin{split}
&JdCov^2(X_1,\dots,X_d;C_2,\dots,C_d)
\\=&C_2\sum_{(i_1,i_2)\in
I^d_2}dCov^2(X_{i_1},X_{i_2})+C_3\sum_{(i_1,i_2,i_3)\in
I^d_3}dCov^2(X_{i_1},X_{i_2},X_{i_3})
\\&+\cdots+C_d \,dCov^2(X_1,\dots,X_d),
\end{split}
\end{equation}
for some nonnegative constants $C_i\geq 0$ with $2\leq i\leq d$.
\end{definition}

Proposition \ref{prop1} below states that JdCov completely
characterizes joint independence among \\$\{X_1,\dots,X_d\}$.
\begin{proposition}\label{prop1}
Suppose $C_i>0$ for $2\leq i\leq d.$ Then
$JdCov^2(X_1,\dots,X_d;C_2,\dots,C_d)=0$ if and only if
$\{X_1,\dots,X_d\}$ are mutually independent.
\end{proposition}

Next we show that by properly
choosing $C_i$s, $JdCov^2(X_1,\dots,X_d;C_2,\dots,C_d)$ has a
relatively simple expression, which does not require the evaluation of $2^d-d-1$
dCov terms in its original definition (\ref{def1}). Specifically, let $C_i=c^{d-i}$
for $c\geq 0$ in the definition of JdCov and denote $JdCov^2(X_1,\dots,X_d;c)=JdCov^2(X_1,\dots,X_d;c^{d-2},c^{d-1},\dots,1)$.
Then, we have the following result.
\begin{proposition}\label{prop2}
For any $c\geq 0$,
\begin{align*}
JdCov^2(X_1,\dots,X_d;c)=\E\left[\prod^{d}_{i=1}\left(
U_i(X_i,X'_i)+c\right)\right]-c^d.
\end{align*}
In particular,
$JdCov^2(X_1,X_2;c)=E[U_1(X_1,X_1')U_2(X_2,X_2')]=dCov^2(X_1,X_2)$.
\end{proposition}


By (\ref{def1}), the dependence measured by JdCov can be decomposed
into the main effect term\\ $\sum_{(i_1,i_2)\in
I^d_2}dCov^2(X_{i_1},X_{i_2})$ quantifying the pairwise
dependence as well as the higher order effect terms
$\sum_{(i_1,i_2,\dots,i_k)\in
I^d_k}dCov^2(X_{i_1},X_{i_2},\dots,X_{i_k})$ quantifying the
multi-way interaction dependence among any $k$-tuples.
The choice of $c$ reflects the relative importance of the main
effect and the higher order effects. For $c\geq 1$, $C_i=c^{d-i}$ is
nonincreasing in $i$. Thus, the larger $c$ we select, the smaller
weights we put on the higher order terms. In particular, we have
$$\lim_{c\rightarrow+\infty}c^{2-d}JdCov^2(X_1,\dots,X_d;c)=\sum_{(i_1,i_2)\in
I^d_2}dCov^2(X_{i_1},X_{i_2}),$$ that is JdCov reduces to the
main effect term as $c\rightarrow+\infty$. We remark that the main
effect term fully characterizes joint dependence in the case of
elliptical distribution and it has been recently used in Yao \textit{et al.}
(2018) to test mutual independence for high-dimensional data. On the
other hand, JdCov becomes the $d$th order dCov as $c\rightarrow 0$, i.e.,
$$\displaystyle \lim_{c\rightarrow
0}JdCov^2(X_1,\dots,X_d;c)=dCov^2(X_1,\dots,X_d).$$
The choice of $c$ depends on the types of interaction dependence of
interest as well as the specific scientific problem, and thus is
left for the user to decide.

It is worth noting that $JdCov^2(X_1,\dots,X_d;c)$ depends on
the scale of $X_i$. To obtain a scale-invariant metric, one can
normalize $U_i$ by the corresponding distance variance. Specifically, when $dCov(X_i):=dCov(X_i,X_i)>0$,
the resulting quantity is given by,
\begin{align*}
JdCov_S^2(X_1,\dots,X_d;c)=\E\left[\prod^{d}_{i=1}\left(
\frac{U_i(X_i,X'_i)}{dCov(X_i)}+c\right)\right]-c^d,
\end{align*}
which is scale-invariant. Another way to obtain a scale-invariant
metric is presented in Section \ref{sec:rank} based on the idea of
rank transformation.

Below we present some basic properties of JdCov, which follow
directly from Proposition \ref{prop0}.
\begin{proposition}\label{prop3}
We have the following properties regarding JdCov:
\begin{enumerate}
\item[(1)] For any $a_i\in\mathbb{R}^{p_i}$, $c_0\in\mathbb{R}$, and orthogonal transformations $A_i\in\mathbb{R}^{p_i\times p_i}$, $JdCov^2(a_1+c_0A_1X_1,\dots,a_d+c_0A_dX_d;|c_0| c)=|c_0|^{d}
JdCov^2(X_1,\dots,X_d;c).$ Moreover, JdCov is invariant to any permutation of $\{X_1,X_2,\dots,X_d\}$.
\item[(2)] For any $a_i\in\mathbb{R}^{p_i}$, $c_i\neq 0$, and orthogonal transformations $A_i\in\mathbb{R}^{p_i\times p_i}$,
$JdCov_S^2(a_1+c_1A_1X_1,\dots,a_d+c_dA_dX_d;c)=JdCov_S^2(X_1,\dots,X_d;c)$.
\end{enumerate}
\end{proposition}

\begin{remark}\label{choice_c}
\rm
A natural question to ask is what should be a data driven way to choose the tuning parameter $c$. Although we leave it for future research, here we present a heuristic idea of choosing $c$. In the discussion below Proposition \ref{prop2}, we pointed out that choosing $c > 1$ (or $< 1$) puts lesser (or higher) weightage on the higher order effects. Note that if the data is Gaussian, testing for the mutual independence of $\{X_1,\dots,X_d\}$ is equivalent to testing for their pairwise independences. In that case, intuitively one should choose a larger ($>1$) value of $c$. If, however, the data is non-Gaussian, it might be of interest to look into higher order dependencies and thus a smaller ($<1$) choice of $c$ makes sense.

To summarize, a heuristic way to choose the tuning parameter $c$ could be :
\begin{equation}
\text{Choose c} \;\;
\begin{cases}
>1 , & \text{if} \; \{X_1,\dots,X_d\}\; \text{are jointly  Gaussian}\\
<1 , & \text{if} \; \{X_1,\dots,X_d\}\; \text{are not jointly Gaussian}.
\end{cases}
\end{equation}

There is a huge literature on testing for joint normality of random vectors. It has been shown that the test based on energy distance is consistent against fixed alternatives (Sz\'{e}kely and Rizzo, 2004) and shows higher empirical power compared to several competing tests (Sz\'{e}kely and Rizzo, 2005; 2013). Suppose $p$ is the p-value of the energy distance based test for joint normality of $\{X_1,\dots,X_d\}$ at level $\alpha$. We expect $c$ to increase (or decrease) from $1$ as $p >$ (or $<$) $\alpha$, so one heuristic choice of $c$ can be
\begin{equation}
c = 1 + \text{sign}(p-\alpha) \times |p-\alpha|^{1/4} \; ,
\end{equation}
where $\text{sign}(x) = 1, 0 \;or -1$ depending on whether $x>0, x=0 \;or\; x<0$. For example, $p=(0.001, 0.03, 0.0499, 0.0501, 0.1, 0.3)$ and $\alpha = 0.05$ yields $c=(0.53, 0.62, 0.9, 1.1, 1.47, 1.71)$.
\end{remark}

\subsection{Distance cumulant and distance characteristic function}\label{sec:cum}
As noted in Streitberg (1990), for $d\geq 4$, the Lancaster interaction measure fails to capture
all possible factorizations of the joint distribution. For example,
it may not vanish if $(X_1,X_2)\independent (X_3,X_4).$
Streitberg (1990) corrected the definition of Lancaster interaction measure
using a more complicated construction, which essentially corresponds
to the cumulant version of dCov in our context. Specifically, Streitberg (1990) proposed a corrected version of Lancaster interaction as follows
\begin{align*}
\widetilde{\Delta}F \, =& \, \displaystyle \sum_{\pi} (-1)^{|\pi| - 1}(|\pi| - 1)!\prod_{D\in \pi}F_D,
\end{align*}
where $\pi$ is a partition of the set \{1,2,\dots,d\}, $|\pi|$ denotes the number of blocks of the partition $\pi$ and $F_D$ denotes the joint distribution of $\{X_i:i\in D\}$.
It has been shown in Streitberg (1990) that $\widetilde{\Delta}F=0$ whenever $F$ is decomposable.
Our definition of joint distance cumulant of $\{X_1,\dots,X_d\}$ below can be viewed as the dCov version of Streitberg's correction.


\begin{definition}
The joint distance cumulant among $\{X_1,\dots,X_d\}$ is defined as
\begin{align}\label{cumulant}
\text{cum}(X_1,\dots,X_d)=\sum_{\pi}(-1)^{|\pi|-1}(|\pi|-1)!\prod_{D\in\pi}\E\left(\prod_{i\in D}U_i(X_i,X_i')\right),
\end{align}
where $\pi$ runs through all partitions of $\{1,2,\dots,d\}$.
\end{definition}
It is not hard to verify that $\text{cum}(X_1,\dots,X_d)=0$
if $\{X_1,\dots,X_d\}$ can be decomposed into two mutually independent
groups say $(X_i)_{i\in \pi_1}$ and $(X_j)_{j\in \pi_2}$ with
$\pi_1$ and $\pi_2$ being a partition of $\{1,2,\dots,d\}$. We further
define the distance characteristic function.

\begin{definition}
The joint distance characteristic function among $\{X_1,\dots,X_d\}$ is
defined as
\begin{equation}
dcf(t_1,\dots,t_d)=\E\left[\exp\left(\I\sum^{d}_{i=1}t_iU_i(X_i,X'_i)\right)\right],
\end{equation}
for $t_1,\dots,t_d\in\mathbb{R}$.
\end{definition}

The following result shows that distance cumulant can be interpreted as the coefficient of the Taylor expansion of the log distance characteristic function.
\begin{proposition}
The joint distance cumulant $cum(X_{i_1},\dots,X_{i_s})$ is given by the
coefficient of $\I^{s}\prod^{s}_{k=1}t_{i_k}$ in the Taylor
expansion of $\log\left\{dcf(t_1,\dots,t_d)\right\}$, where
$\{i_1,\dots,i_s\}$ is any subset of $\{1,2,\dots,d\}$ with $s\leq d.$
\end{proposition}

Our next result indicates that the mutual independence among
$\{X_1,\dots,X_d\}$ is equivalent to the mutual independence among
$\{U_1(X_1,X_1'),\dots,U_d(X_d,X_d')\}$.
\begin{proposition}\label{prop-cha}
The random variables $\{X_1,\dots,X_d\}$
are mutually independent if and only if\\ $dcf(t_1,\dots,t_d)=\prod_{i=1}^{d}dcf(t_i)$ for $t_i$ almost everywhere, where $dcf(t_i)=\E[\exp\{\I t_iU_i(X_i,X_i')\}]$.
\end{proposition}



\subsection{Rank-based metrics}\label{sec:rank} In this subsection, we briefly discuss
the concept of rank-based distance measures. For simplicity, we
assume that $X_i$s are all univariate and remark that our definition can be generalized to the case where $X_i$s are random vectors without essential difficulty.
The basic idea here is to apply the monotonic transformation based on the marginal
distribution functions to each $X_j$, and then use the
dCov or JdCov to quantify the interaction and joint
dependence of the coordinates after transformation. Therefore it can be
viewed as the counterpart of Spearman's rho to dCov or JdCov.

Let $F_j$ be the marginal distribution function for $X_j$. The
squared rank dCov and JdCov among $\{X_1,\dots,X_d\}$ are defined
respectively as
\begin{align*}
&
dCov^2_R(X_1,\dots,X_d)=dCov^2(F_1(X_1),\dots,F_d(X_d)),\\
&
JdCov_R^2(X_1,\dots,X_d;c)=JdCov^2(F_1(X_1),\dots,F_d(X_d);c).
\end{align*}
The rank-based dependence metrics enjoy a few appealing features:
(1) they are invariant to monotonic component wise transformations;
(2) they are more robust to outliers and heavy tail of the
distribution; (3) their existence require very weak moment
assumption on the components of $\mathcal{X}$. In Section \ref{sec:num}, we shall compare the finite sample performance of $\text{JdCov}^2_R$ with that of JdCov and $\text{JdCov}_S$.\\

\begin{table}[!h]
\centering
\caption*{{\bf Comparison of various distance metrics for measuring joint dependence of $d\geq 2$ random vectors of arbitrary dimensions :}}
\begin{tabular}{|c|c|c|c| }
\toprule
Distance metrics & Complete characterization   & Permutation  & Scale  \\
 & of joint independence & invariance & invariance \\
\hline
\multirow{1}{*}{dHSIC}& \checkmark & \checkmark & $\pmb{ \times}$ (for fixed bandwidth) \\
\hline
\multirow{1}{*}{$T_{MT}$} & \checkmark   & $\pmb{ \times}$ & $\pmb{ \times}$ \\
\hline
\multirow{1}{*}{High order $dCov$}& $\pmb{ \times}$ (Captures Lancaster interactions) & \checkmark & $\pmb{ \times}$ \\
\hline
\multirow{1}{*}{$JdCov$}& \checkmark  & \checkmark & $\pmb{ \times}$ \\
\hline
\multirow{1}{*}{$JdCov_S$}& \checkmark & \checkmark & \checkmark \\
\hline
\multirow{1}{*}{$JdCov_R$}& \checkmark & \checkmark & \checkmark \\

\toprule

\end{tabular}
\end{table}

\section{Estimation}\label{sec:est}
We now turn to the estimation of the joint dependence metrics.
Given $n$ samples $\{\X_j\}^{n}_{j=1}$ with $\X_j=(X_{j1},\dots,X_{jd})$, we consider
the plug-in estimators based on the V-statistics as well as their bias-corrected versions to be described below.
Denote by $\hat{f}_i(t_i)=n^{-1}\sum_{j=1}^ne^{\imath \langle t_i,X_{ji}\rangle}$ the empirical characteristic function for $X_i$.

\subsection{Plug-in estimators}\label{est}
For $1\leq k,l\leq n$, let
$\widehat{U}_{i}(k,l)=n^{-1}\sum_{v=1}^{n}|X_{ki}-X_{vi}|+n^{-1}\sum_{u=1}^{n}|X_{ui}-X_{li}|-|X_{ki}-X_{li}|-n^{-2}\sum_{u,v=1}^{n}|X_{ui}-X_{vi}|$ be the sample estimate of $U_i(X_{ki},X_{li})$.
The V-statistic type estimators for dCov, JdCov and its scale-invariant version are defined
respectively as,
\begin{align}
&\widehat{dCov^2}(X_1,\dots,X_d)=\frac{1}{n^2}\sum^{n}_{k,l=1}\prod_{i=1}^{d}\widehat{U}_{i}(k,l)^2,
\\&\widehat{JdCov^2}(X_1,\dots,X_d;c))=\frac{1}{n^2}\sum_{k,l=1}^{n}\prod_{i=1}^{d}\left(\widehat{U}_{i}(k,l)+c\right)-c^d,
\label{eq-bias1}
\\&\widehat{JdCov^2_S}(X_1,\dots,X_d;c)=\frac{1}{n^2}\sum_{k,l=1}^{n}\prod_{i=1}^{d}\left(\frac{\widehat{U}_{i}(k,l)}{\widehat{dCov}(X_i)}+c\right)-c^d, \label{eq-bias2}
\end{align}
where
$\widehat{dCov^2}(X_i)=n^{-2}\sum_{k,l=1}^{n}\widehat{U}_{i}(k,l)^2$
is the sample (squared) dCov. The following lemma shows that the V-statistic type estimators are equivalent to the plug-in estimators by replacing the characteristic functions and the expectation in the definitions of dCov and JdCov with their sample counterparts.
\begin{lemma}\label{lemma2}
The sample (squared) dCov can be rewritten as,
\begin{align}\label{eq-lemma2}
\widehat{dCov^2}(X_1,\dots,X_d)=\int_{\mathbb{R}^{p_0}}\left|\frac{1}{n}\sum^{n}_{k=1}\left[\prod_{i=1}^{d}(\hat{f}_i(t_i)-e^{\I\langle
t_i,X_{ki}\rangle})\right]\right|^2dw.
\end{align}
Moreover, we have
\begin{equation}\label{eq2-lemma2}
\begin{split}
&\widehat{JdCov^2}(X_1,\dots,X_d;c)
\\=&c^{d-2}\sum_{(i_1,i_2)\in
I^d_2}\widehat{dCov^2}(X_{i_1},X_{i_2})+c^{d-3}\sum_{(i_1,i_2,i_3)\in
I^d_3}\widehat{dCov^2}(X_{i_1},X_{i_2},X_{i_3})
\\&+\cdots+\widehat{dCov^2}(X_1,\dots,X_d).
\end{split}
\end{equation}
\end{lemma}

\begin{remark}
{\rm Consider the univariate case where $p_i=1$ for all $1\leq i\leq d.$
Let $\widehat{F}_i$ be the empirical distribution based on
$\{X_{ji}\}^{n}_{j=1}$ and define $Z_{ji}=\widehat{F}_i(X_{ji})$.
Then, the rank-based metrics defined in Section \ref{sec:rank} can be estimated in a similar way by
replacing $X_{ji}$ with $Z_{ji}$ in the definitions of the above
estimators. }
\end{remark}

\begin{remark}
{\rm
The distance cumulant can be estimated by
$$\widehat{\text{cum}}(X_1,\dots,X_d)=\sum_{\pi}(-1)^{|\pi|-1}(|\pi|-1)!\prod_{D\in\pi}\left\{\frac{1}{n^2}\sum_{k,l=1}^{n}\left(\prod_{i\in D}\widehat{U}_i(k,l)\right)\right\}.$$
However, the combinatorial nature of distance cumulant implies that detecting interactions of higher order requires significantly more costly
computation.
}
\end{remark}


We study the asymptotic properties of the V-statistic type
estimators under suitable moment assumptions.
\begin{assumption}\label{ass1}
Suppose for any subset $S$ of $\{1,2,\dots,d\}$ with $|S|\geq 2$, there exists a partition $S=S_1\cup S_2$ such that $\E \prod_{i\in S_{1}}|X_i|<\infty$ and $\E \prod_{i\in S_{2}}|X_i|<\infty$.
\end{assumption}

\begin{proposition}\label{prop5}
Under Assumption \ref{ass1} 
, we have as $n \to \infty$,
\begin{align*}
&\widehat{dCov^2}(X_1, \cdots , X_d) \overset{a.s}{\longrightarrow} dCov^2(X_1, \cdots , X_d),\\
&\widehat{JdCov^2}(X_1, \cdots , X_d; c) \overset{a.s}{\longrightarrow} JdCov^2(X_1,\dots,X_d ;c),\\
&\widehat{JdCov^2_S}(X_1, \cdots , X_d; c)
\overset{a.s}{\longrightarrow} JdCov^2_S(X_1,\dots,X_d ;c),
\end{align*}
where $``\overset{a.s}{\longrightarrow}"$ denotes the almost sure convergence.
\end{proposition}
When $d=2$, Assumption \ref{prop5} reduces to the condition that $\E|X_1|<\infty$ and $\E|X_2|<\infty$ in Theorem 2 of Sz\'{e}kely et al. (2007).
Suppose $X_i$s are mutually independent. Then Assumption \ref{prop5} is fulfilled provided that $\E|X_i|<\infty$ for all $i.$
More generally, if $E|X_i|^{\lfloor (d+1)/2 \rfloor}<\infty$ for $1\leq i\leq d$, then Assumption \ref{prop5} is satisfied.

Let $\Gamma(\cdot)$ denote a complex-valued zero mean Gaussian
random process with the covariance function $R(t,t')=\prod_{i=1}^d
\big(f_i(t_i - t_{i}') - f_i(t_i) f_i (-t_{i}')\big)$, where
$t=(t_1, t_2,\dots,t_d), t'=(t_{1}',t_{2}',\dots,t_{d}') \in
\mathbb{R}^{p_1} \times \mathbb{R}^{p_2} \times \cdots \times
\mathbb{R}^{p_d}$.

\begin{proposition}\label{prop6}
Suppose $X_1, X_2,\dots,X_d$ are mutually independent, and $\E|X_i|<\infty$
for $1\leq i\leq d.$ Then we have
\[n\widehat{dcov^2}(X_1, X_2, \cdots
, X_d) \overset{d}{\longrightarrow} \Arrowvert
\Gamma\Arrowvert^2=\sum^{+\infty}_{j=1}\lambda_jZ_j^2,\] where
$||\Gamma||^2=\int \Gamma(t_1,t_2,\dots,t_d)^2 dw$, $Z_j \overset{i.i.d}{\sim}
N(0,1)$ and $\lambda_j>0$ depends on the distribution of
$\mathcal{X}$. As a consequence, we have
$$n\widehat{Jdcov^2}(X_1, X_2,\cdots
, X_d ;c) \overset{d}{\longrightarrow}
\sum^{+\infty}_{j=1}\lambda_j'Z_j^2,$$ with $\lambda_j'>0$ and $Z_j\overset{i.i.d}{\sim}
N(0,1)$.
\end{proposition}
Proposition \ref{prop6} shows that both $\widehat{dcov^2}$ and $\widehat{Jdcov^2}$ converge to weighted sum of chi-squared random variables, where the weights depend on the marginal characteristic functions
in a complicated way. Since the limiting distribution is non-pivotal, we will introduce a bootstrap procedure to approximate their sampling distributions in the next section.

It has been pointed out in the literature that the computational complexity of dCov is $O(n^2)$ if it is implemented directly according to its definition. The computational cost of the V-statistic type estimators and the bias-corrected estimators for JdCov are both of the order $O(n^2p_0)$.

\subsection{Bias-corrected estimators}\label{bias corr}
It is well known that V-statistic leads to biased estimation. To remove the bias, one can construct an estimator for the $d$th order dCov based on a
$d$th order U-statistic. However, the computational complexity for the $d$th order U-statistic is of the order $O(dn^d)$, which is computationally prohibitive when $n$
and $d$ are both large. Adopting the $\mathcal{U}$-centering idea in Sz\'{e}kely and Rizzo (2014), we propose bias-corrected estimators which do not bring extra computational cost as compared to the
plug-in estimators. Specifically, for $1\leq
i\leq d$, we define the $\mathcal{U}$-centered version of $|X_{ki}-X_{li}|$ as
\begin{align*}
\widetilde{U}_{i}(k,l)=&\frac{1}{n-2}\sum^{n}_{u=1}|X_{ui}-X_{li}|+\frac{1}{n-2}\sum^{n}_{v=1}|X_{ki}-X_{vi}|-|X_{ki}-X_{li}|
\\&-\frac{1}{(n-1)(n-2)}\sum_{u,v=1}^{n}|X_{ui}-X_{vi}|
\end{align*}
when $k\neq l$, and $\widetilde{U}_{i}(k,l)=0$ when $k=l$. One can verify
that $\sum_{v\neq k}\widetilde{U}_{i}(k,v)=\sum_{u\neq
l}\widetilde{U}_{i}(u,l)=0$, which mimics the double-centered property $\E[U_i(X_i,X'_i)|X_i]=\E[U_i(X_i,X'_i)|X_i']=0$ for its population counterpart.
Let $\widetilde{dCov^2}(X_i,X_j)=\sum_{k\neq
l}\widetilde{U}_{i}(k,l)\widetilde{U}_{j}(k,l)/\{n(n-3)\}$ and write $\widetilde{dCov}(X_i)=\widetilde{dCov}(X_i,X_i).$
We define the bias-corrected estimators as,
\begin{align*}
&\widetilde{JdCov^2}(X_1,\dots,X_d;c)=\frac{1}{n(n-3)}\sum_{k,l=1}^n\prod_{i=1}^{d}\left(\widetilde{U}_{i}(k,l)+c\right)-\frac{n}{n-3}c^d,
\\&\widetilde{JdCov_S^2}(X_1,\dots,X_d;c)=\frac{1}{n(n-3)}\sum_{k,l=1}^n\prod_{i=1}^{d}\left(\frac{\widetilde{U}_{i}(k,l)}{\widetilde{dCov}(X_i)}+c\right)-\frac{n}{n-3}c^d.
\end{align*}
Direct calculation yields that
\begin{align}\label{jdov1}
\widetilde{JdCov^2}(X_1,\dots,X_n;c)=c^{d-2}\sum_{(i,j)\in
I^d_2}\widetilde{dCov}^2(X_i,X_j)+\text{higher order terms}.
\end{align}
It has been shown in Proposition 1 of Sz\'{e}kely and Rizzo (2014) that $\widetilde{dCov}^2(X_i,X_j)$ is an unbiased estimator for $dCov^2(X_i,X_j)$.
In the supplementary material, we provide an alternative proof which simplifies the arguments in Sz\'{e}kely and Rizzo (2014).
Our argument relies on a new decomposition of $\widetilde{U}_{i}(k,l)$, which provides some insights on the
$\mathcal{U}$-centering idea. See Lemma \ref{lemma5} and Proposition \ref{prop9} in the supplementary material. In view of (\ref{jdov1}) and Proposition \ref{prop9}, the main effect in $JdCov^2(X_1,\dots,X_n;c)$ can be unbiasedly estimated
by the main effect of $\widetilde{JdCov^2}(X_1,\dots,X_n;c)$. However, it seems very challenging to study the impact of
$\mathcal{U}$-centering on the bias of the high order effect terms. We shall leave this problem to our future research.

\section{Testing for joint independence}\label{sec:infer}
In this section, we consider the problem of testing the null hypothesis
\begin{equation}
H_0: X_1,\dots,X_d \text{ are mutually independent}
\end{equation}
against the alternative $H_A:$ negation of $H_0$. For the purpose of illustration, we use $n\widehat{JdCov^2}$ as our test statistic and set
\begin{equation}
\phi_n(\X_1,\dots,\X_n):= \begin{cases}
1 & if \quad n\widehat{JdCov^2}(X_1, \dots, X_d) > c_n\, ,\\
0 & if \quad n\widehat{JdCov^2}(X_1, \dots, X_d) \leq c_n \, ,
\end{cases}
\end{equation}
where the threshold $c_n$ remains to be chosen. Consequently, we define a decision rule as follows: reject $H_0$ if $\phi_n=1$ and fail to reject $H_0$ if $\phi_n=0$.

Below we introduce a bootstrap procedure to approximate the sampling
distribution of $n\widehat{JdCov}$ under $H_0$. Let $\widehat{F}_i$ be the empirical distribution function based on
the data points $\{X_{ji}\}^{n}_{j=1}$. Conditional on the original sample,
we define $\X_j^*=(X_{j1}^*,\dots,X_{jd}^*)$, where $X_{ji}^*$ are
generated independently from $\widehat{F}_i$ for $1\leq i\leq d$.
Let $\{\X_j^*\}^{n}_{j=1}$ be $n$ bootstrap samples. Then we can compute the bootstrap statistics
$\widehat{dCov^2}^*$ and $\widehat{JdCov^2}^*$ in the same way as $\widehat{dCov^2}$ and $\widehat{JdCov^2}$ based on $\{\X_j^*\}^{n}_{j=1}$.
In particular, we note that the bootstrap version of the $d$th order dCov is given by
$$n\widehat{dCov^2}^*(X_1,\dots,X_d)=\Arrowvert \Gamma^*_n \Arrowvert^2=\int \Gamma^*_n(t_1,\dots,t_d)^2dw,$$
where
$$\Gamma^*_n (t) \, = \,n^{-1/2}\displaystyle \sum_{j=1}^n \prod_{i=1}^d (\hat{f^*_i} (t_i) - e^{\imath \langle t_i,X^*_{ji}\rangle}).$$
Denote by $``\overset{d^*}{\longrightarrow}"$ the weak convergence in the bootstrap world conditional on the original sample $\{\X_j\}^{n}_{j=1}$.
\begin{proposition}\label{prop10}
Suppose $\E|X_i|<\infty$ for $1\leq i\leq d$. Then 
\begin{align*}
&n\widehat{dCov^2}^*(X_1,\dots,X_d) \, \overset{d^*}{\longrightarrow}\sum^{+\infty}_{j=1}\lambda_jZ_j^2,
\\&n\widehat{JdCov^2}^*(X_1,\dots,X_d) \, \overset{d^*}{\longrightarrow} \,\sum^{+\infty}_{j=1}\lambda_j'Z_j^2,
\end{align*}
almost surely as $n \to \infty$.
\end{proposition}
Proposition \ref{prop10} shows that the bootstrap statistic is able to imitate the limiting distribution of the test statistic.
Thus, we shall choose $c_n$ to be the $1-\alpha$ quantile of the distribution of $n \widehat{JdCov^2}^*$ conditional on the sample $\{\X_j\}^{n}_{j=1}$. The validity of the bootstrap-assisted test can be justified as follows.



\begin{proposition}\label{prop11}
For all $\alpha \in (0,1)$, the $\alpha$-level bootstrap-assisted test has asymptotic level $\alpha$ when testing $H_0$ against $H_A$. In other words, under $H_0$, $\displaystyle \limsup_{n \to \infty} P\left(\, \phi_n(\X_1,\dots,\X_n)=1 \,\right)=\alpha \;$.
\end{proposition}

\begin{proposition}\label{prop12}
For all $\alpha \in (0,1)$, the $\alpha$-level bootstrap-assisted test is consistent when testing $H_0$ against $H_A$. In other words, under $H_A$, $\displaystyle \lim_{n \to \infty} P\left(\, \phi_n(\X_1,\dots,\X_n)=1\,\right)=1  \;$.
\end{proposition}



\section{Numerical studies}\label{sec:num}
We investigate the finite sample performance of the proposed methods. Our first goal is to test the joint independence among the variables $\{X_1,\dots,X_d\}$ using the new dependence metrics, and compare the performance with some existing alternatives in the literature in terms of size and power.
Throughout the simulation, we set $c=0.5,1,2$ in JdCov and implement the bootstrap-assisted test based on the bias-corrected estimators. We compare our tests with the dHSIC-based test in Pfister et al. (2018) and the
mutual independence test proposed in Matteson and Tsay (2017), which is defined as
\begin{align}\label{eq-MT16}
T_{MT}:=\sum_{i=1}^{d-1}dCov^2(X_i,X_{(i+1):d}),
\end{align}
where $X_{(i+1):d}=\{X_{i+1},X_{i+2},\dots,X_d\}$. We consider both Gaussian and non-Gaussian distributions and study the following models, motivated from Sejdinovic et al. (2013) and Yao et al. (2018).


\begin{example}\label{eg:size}[Gaussian copula model]
{\rm The data $\mathbf{X}=(X_1,\dots,X_d)$ are generated as follows:
\begin{enumerate}
\item $\mathbf{X} \sim N(0, I_d)$;
\item $\mathbf{X}=Z^{1/3}$ and $Z \sim N(0, I_d)$;
\item $\mathbf{X}=Z^3$ and $Z \sim N(0, I_d)$.
\end{enumerate}
}

\end{example}


\begin{example}\label{eg:power_normal}[Multivariate Gaussian model]
{\rm
The data $\mathbf{X}=(X_1,\dots,X_d)$ are generated from the multivariate normal distribution with the following three covariance matrices $\Sigma=(\sigma_{ij}(\rho))_{i,j=1}^d$ with $\rho=0.25$:
\begin{enumerate}
\item AR(1): $\sigma_{ij}=\rho^{|i-j|}$ for all $i,j \in \{1,\dots,d\}$;
\item Banded: $\sigma_{ii}=1$ for $i=1,\dots,d$; $\sigma_{ij}=\rho$ if $1 \leq  |i-j| \leq  2$ and $\sigma_{ij}=0$ otherwise;
\item Block: Define $\Sigma_{\text{block}}=(\sigma_{ij})^{5}_{i,j=1}$ with $\sigma_{ii}=1$ and $\sigma_{ij}=\rho$ if $i \ne j$. Let $\Sigma=I_{\lfloor d/5 \rfloor} \otimes \Sigma_{\text{block}}$,
where $\otimes$ denotes the Kronecker product.
\end{enumerate}
}
\end{example}


\begin{example}\label{eg:mutua1-2}
{\rm
The data $\mathbf{X}=(X, Y, Z)$ are generated as follows:
\begin{enumerate}
\item $X, Y \overset{i.i.d}{\sim} N(0,1)$, $Z=$\,sign$(X Y)\, W$, where $W$ follows an exponential distribution with mean $\sqrt{2}$;
\item $X, Y$ are independent Bernoulli random variables with the success probability $0.5$, and $Z = \mathbf{1}\{X=Y\}$.
\end{enumerate}

}
\end{example}


\begin{example}\label{eg:mutua1-3}
{\rm
In this example, we consider a triplet of random vectors $(X, Y, Z)$ on $\mathbb{R}^p \times \mathbb{R}^p \times \mathbb{R}^p$, with $X, Y \overset{i.i.d}{\sim} N(0,I_p)$. We focus on the following cases :
\begin{enumerate}
\item $Z_1 = \text{sign}(X_1 Y_1)\, W$ and $Z_{2:p} \sim N(0,I_{p-1})$, where $W$ follows an exponential distribution with mean $\sqrt{2}$;

\item $Z_{2:p} \sim N(0,I_{p-1})$ and
\[Z_1 = \begin{cases}
X_1^2 + \epsilon, & \text{with probability} \, 1/3, \\
Y_1^2 + \epsilon, & \text{with probability} \, 1/3, \\
X_1 Y_1 + \epsilon, & \text{with probability} \, 1/3, \\
\end{cases}\]  where $\epsilon \sim U(-1,1)$.


\end{enumerate}
}
\end{example}

We conduct tests for joint independence among the random variables described in the above examples. For each example, we draw 1000 simulated datasets and perform tests of joint independence with 500 bootstrap resamples. We try small and moderate sample sizes, i.e., $n=50, 100$ or $200$. Figure \ref{fig:5.1,5.2} and Figure \ref{fig:5.3,5.4} display the proportion of rejections (out of 1000 simulation runs) for the five different tests, based on the statistics $\widetilde{JdCov^2}$, $\widetilde{JdCov^2_S}$, $\widetilde{JdCov^2_R}$, dHSIC and $T_{MT}$. The detailed figures are reported in Tables \ref{table:tb} and $\ref{table:tb1}$ in the supplementary materials.

In Example \ref{eg:size}, the data generating scheme suggests that the variables are jointly independent. The plots in Figure \ref{fig:5.1,5.2}\, show that all the five tests perform more or less equally well in examples \ref{eg:size}.1 and \ref{eg:size}.2, and the rejection probabilities are quite close to the $10\%$ or $5\%$ nominal level. In Example \ref{eg:size}.3, the tests based on our proposed statistics show greater conformation of the empirical size to the actual size of the test than $T_{MT}$. In Example $\ref{eg:power_normal}$, the tests based on $\widetilde{JdCov^2}$, $\widetilde{JdCov^2_S}$ and $\widetilde{JdCov^2_R}$ as well as $T_{MT}$ significantly outperform the dHSIC-based test. Note that the empirical power becomes higher when $c$ increases to 2. From Figure \ref{fig:5.3,5.4}, we observe that in Example $\ref{eg:mutua1-2}$ all the tests perform very well in the second case. However, in the first case, our tests and the dHSIC-based test deliver higher power as compared to $T_{MT}$. Finally, in Example \ref{eg:mutua1-3}, we allow $X,Y,Z$ to be random vectors with dimension $p=5,10$. The rejection probabilities for each of the five tests increase with $n$, and the proposed tests provide better performances in comparison with the other two competitors. In particular, the test based on $\widetilde{JdCov^2_S}$ outperforms all the others in a majority of the cases. In Examples $\ref{eg:mutua1-2}$ and $\ref{eg:mutua1-3}$, the power becomes higher when $c$ decreases to 0.5. These results are consistent with our statistical intuition and the discussions in Section \ref{sec:jdov}. For the Gaussian copula model, only the main effect term matters, so a larger $c$ is preferable. For non-Gaussian models, the high order terms kick in and hence a smaller $c$ may lead to higher power.

\begin{remark}\label{bias corr_emp}
\rm
We have considered U-statistic type estimators of $JdCov^2$, $JdCov_S^2$ and $JdCov_R^2$ so far in all the above computations, as they remove the bias due to the main effects (see Section \ref{bias corr}). However it might be interesting to see if the bias correction has any empirical impact. We conduct tests for joint independence of the random variables in some of the above examples, this time using the V-statistic type estimators (described in Section \ref{est}). Table \ref{table:tb_bias} (in the supplementary materials) shows the proportion of rejections (out of 1000 simulation runs) for the tests based on $\widehat{JdCov^2}$, $\widehat{JdCov^2_S}$ and $\widehat{JdCov^2_R}$, setting $c=1$. The results indicate that use of the bias corrected estimators lead to greater conformation of the empirical size to the actual size of the test (in Example \ref{eg:size}), and slightly better power in Example \ref{eg:mutua1-2}.

\end{remark}

\begin{remark}\label{sim_c}
\rm
In connection to the heuristic idea discussed in Remark \ref{choice_c} about choosing the tuning parameter $c$, we conduct tests for joint independence of the random variables in all the above examples, choosing $c$ in that way. Table \ref{table:tb3} (in the supplementary materials) presents the proportion of rejections for the proposed tests and the values of $c$ for each example, averaged over the 1000 simulated datasets. The plots in Figure \ref{fig:5.1,5.2} and Figure \ref{fig:5.3,5.4} reveal some interesting features. In Example \ref{eg:power_normal} we have Gaussian data, so a larger $c$ is preferable. Clearly the proportion of rejections are a little higher (or lower) in most of the cases when we choose $c$ in the data-driven way ($c$ turns out to be around 1.6 or 1.7), than when $c$ is subjectively chosen to be  0.5 (or 2). On the contrary, in Example \ref{eg:mutua1-2}, the data is non-Gaussian and a smaller $c$ is preferable. Evidently choosing $c$ in the data-driven way leads to nearly equally good power compared to when $c=0.5$, and higher power compared to when $c=2$.

\end{remark}

\section{Application to causal inference}\label{sec:dag}
\subsection{Model diagnostic checking for Directed Acyclic Graph (DAG)}
We employ the proposed metrics to perform model selection in causal inference
which is based on the joint independence testing of the residuals from the fitted structural
equation models. Specifically, given
a candidate DAG $\mathcal{G}$, we let $\text{Par}(j)$ denote the
index associated with the parents of the $j$th node. Following Peters et al. (2014) and B\"uhlmann et al. (2014), we consider the structural equation models with additive components
\begin{align}\label{additive}
X_j = \displaystyle \sum_{k \in \text{Par}(j)} f_{j,k}(X_k) + \epsilon_j \, , \; j=1,2,\dots,d,
\end{align}
where the noise variables $\epsilon_1, \dots \, ,\epsilon_d$ are jointly independent variables. Given $n$
observations $\{\X_i\}^{n}_{i=1}$ with
$\X_i=(X_{i1},\dots,X_{id})$, we use generalized additive regression (Wood and Augustin, 2002) to regress $X_j$ on all its parents
$\{X_k,k\in\text{Par}(j)\}$ and denote the resulting residuals by
\begin{align*}
\hat{\epsilon}_{ij}=X_{ij}\, - \displaystyle \sum_{k \in \text{Par}(j)}\hat{f}_{j,k}(X_{ik}),\quad 1\leq j\leq d,\quad 1\leq i\leq n,
\end{align*}

\begin{figure}[H]
    \centering
    \subfloat[]{\includegraphics[width=0.44\textwidth]{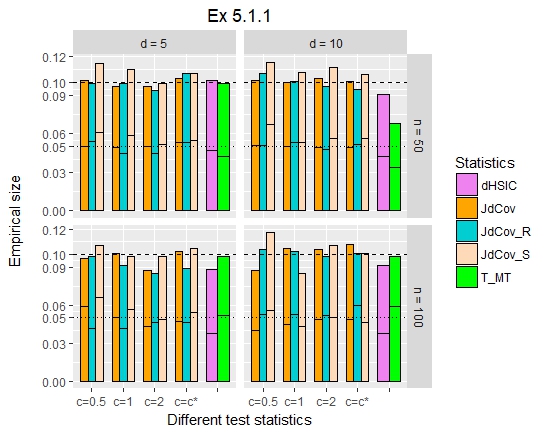}}
    \hfill
    \subfloat[]{\includegraphics[width=0.44\textwidth]{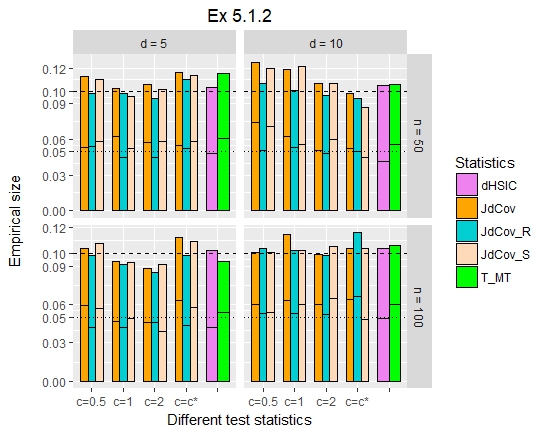}}
    \vspace{0.5cm}
    \subfloat[]{\includegraphics[width=0.44\textwidth]{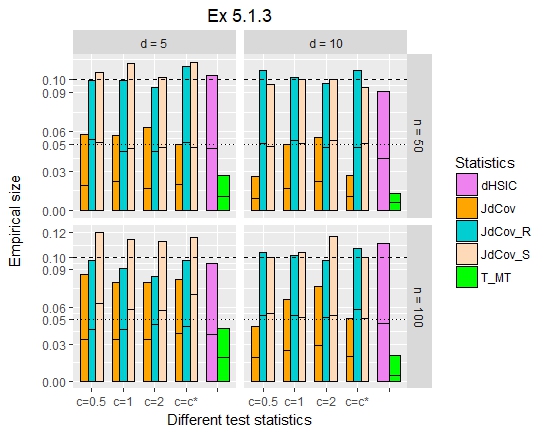}}
    \hfill
    \subfloat[]{\includegraphics[width=0.44\textwidth]{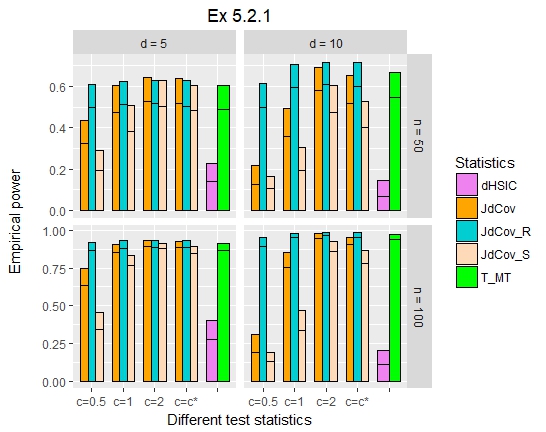}}
    \vspace{0.5cm}
    \subfloat[]{\includegraphics[width=0.44\textwidth]{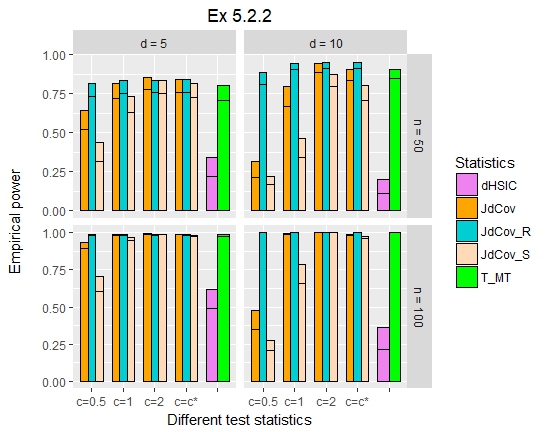}}
    \hfill
    \subfloat[]{\includegraphics[width=0.44\textwidth]{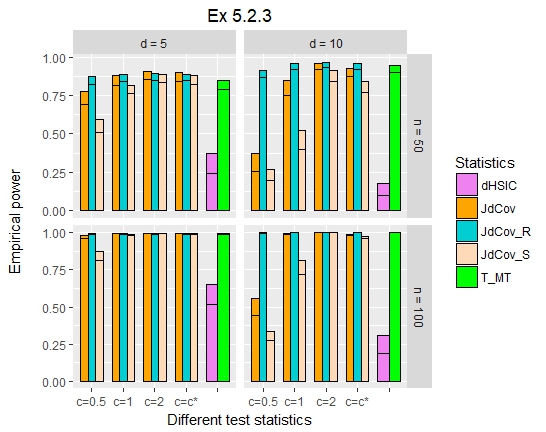}}

    \caption{Figures showing the empirical size and power for the different tests statistics in Examples \ref{eg:size} and \ref{eg:power_normal}. $c^*$ denotes the data-driven choice of $c$. The vertical height of a bar and a line on a bar stand for the empirical size or power at levels $\alpha=0.1$ or $\alpha=0.05$, respectively.}\label{fig:5.1,5.2}
\end{figure}

\begin{figure}[H]
    \centering
    \subfloat[]{\includegraphics[width=0.44\textwidth]{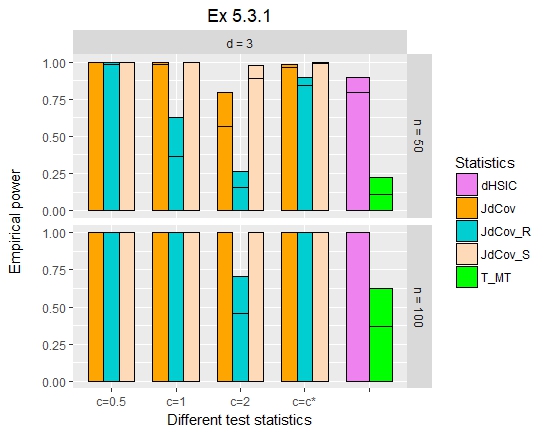}}
    \hfill
    \subfloat[]{\includegraphics[width=0.44\textwidth]{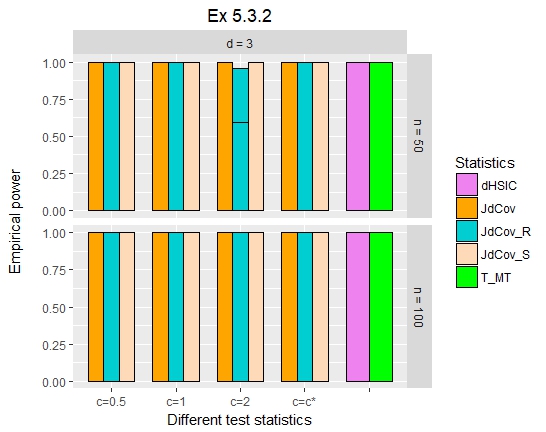}}
    \vspace{0.5cm}
    \subfloat[]{\includegraphics[width=0.44\textwidth]{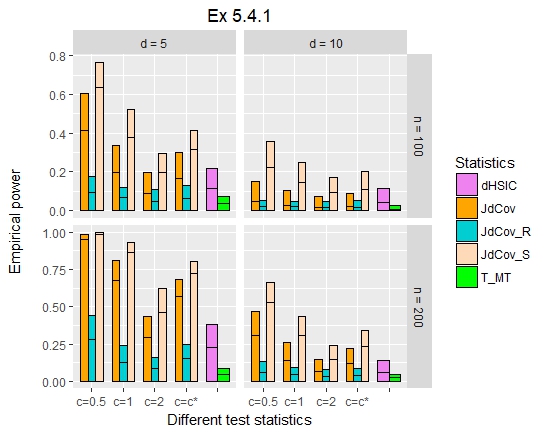}}
    \hfill
    \subfloat[]{\includegraphics[width=0.44\textwidth]{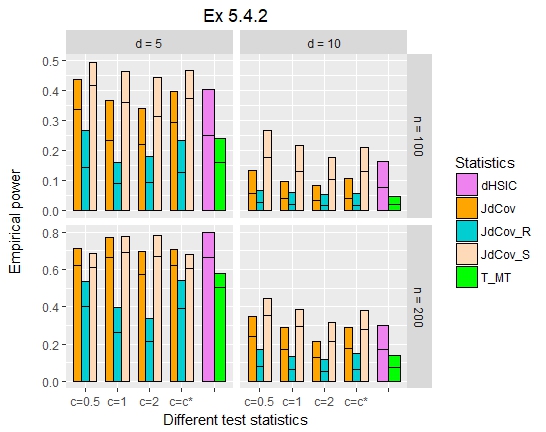}}

    \caption{Figures showing the empirical power for the different tests statistics in Examples \ref{eg:mutua1-2} and \ref{eg:mutua1-3}. $c^*$ denotes the data-driven choice of $c$. The vertical height of a bar and a line on a bar stand for the empirical power at levels $\alpha=0.1$ or $\alpha=0.05$, respectively.}\label{fig:5.3,5.4}
\end{figure}

where $\hat{f}_{j,k}$ is the B-spline estimator for $f_{j,k}$. To check the goodness of fit of $\mathcal{G}$, we test
the joint independence of the residuals. Let
$T_n$ be the statistic (e.g. $\widetilde{JdCov^2}$, $\widetilde{JdCov^2_S}$ or $\widetilde{JdCov^2_R}$) to test the joint dependence of $(\epsilon_1,\dots,\epsilon_d)$
constructed based on the fitted residuals $\hat{\epsilon}_i=(\hat{\epsilon}_{i1},\dots,\hat{\epsilon}_{id})$ for $1\leq i\leq n.$ Following the idea presented in Sen and Sen (2014), it seems that $T_n$ might have a limiting distribution different from the one mentioned in Proposition $\ref{prop6}$. So to approximate the sampling distribution of $T_n$, we introduce the
following residual bootstrap procedure.
\begin{enumerate}
  \item Randomly sample $\epsilon^*_j=(\epsilon^*_{1j},\dots,\epsilon^*_{n j})$ with replacement from the residuals $\{\hat{\epsilon}_{1j},\dots,\hat{\epsilon}_{n j}\}$, $1 \leq j \leq d$. Construct the bootstrap sample $X_{ij}^*=\sum_{k\in \text{Par}(j)}\hat{f}_{j,k}(X_{ik})+\epsilon_{ij}^*$.
  \item Based on the bootstrap sample $\{\X_i^*\}^{n}_{i=1}$ with $\X_i^*=(X_{i1}^*,\dots,X_{id}^*)$, estimate $f_{j,k}$ for $k\in \text{Par}(j)$, and denote the corresponding residuals by $\hat{\epsilon}^*_{ij}.$
  \item Calculate the bootstrap statistic $T_n^*$ based on $\{\hat{\epsilon}^*_{ij}\}$.
  \item Repeat the above steps $B$ times and let $\{T_{b,n}^*\}^{B}_{b=1}$ be the corresponding values of the bootstrap statistics. The $p$-value is given by $B^{-1}\sum_{b=1}^{B}\{T_{b,n}^* > T_n\}$.
\end{enumerate}
Pfister et al. (2018) proposed to bootstrap the residuals directly and used the bootstrapped residuals to construct the test statistic. In contrast, we suggest the use of the above residual bootstrap
to capture the estimation effect caused by replacing $f_{j,k}$ with the estimate $\hat{f}_{j,k}$.

\subsection{Real data example}\label{6.2}
We now apply the model diagnostic checking procedure for DAG to one real world dataset. A population of women who were at least 21 years old, of Pima Indian heritage and living near Phoenix, Arizona, was tested for diabetes according to World Health Organization criteria. The data were collected by the US National Institute of Diabetes and Digestive and Kidney Diseases. We downloaded the data from \url{https://archive.ics.uci.edu/ml/datasets/Pima+Indians+Diabetes}. We focus only on the following five variables : Age, Body Mass Index (BMI), 2-Hour Serum Insulin (SI), Plasma Glucose Concentration (glu) and Diastolic Blood Pressure (DBP). Further, we only selected the instances with non-zero values, as it seems that zero values encode missing data. This yields $n=392$ samples.

Now, age is likely to affect all the other variables (but of course not the other way round). Moreover, serum insulin also has plausible causal effects on BMI and plasma glucose concentration. We try to determine the correct causal structure out of $48$ candidate DAG models and perform model diagnostic checking for each of the $48$ models, as illustrated in Section 6.1. We first center each of the variables and scale them so that $l_2$ norm of each of the variables is $\sqrt{n}$. We perform the mutual independence test of residuals based on the statistics $\widetilde{JdCov^2}$, $\widetilde{JdCov^2_S}$ and $\widetilde{JdCov^2_R}$ with $c =1$, and compare with the bootstrap-assisted version of the dHSIC-based test proposed in Pfister et al. (2018) and $T_{MT}$. For each of the tests, we implement the residual bootstrap to obtain the p-value with $B=1000.$
Figure \ref{fig1} shows the selected DAG models corresponding to the largest p-values from each of the five tests.

\begin{figure}[h]
  \centering
  \subfloat[$\widetilde{JdCov^2}$, $\widetilde{JdCov^2_S}$, $\widetilde{JdCov^2_R}$ and $T_{MT}$]{\includegraphics[width=0.4\textwidth]{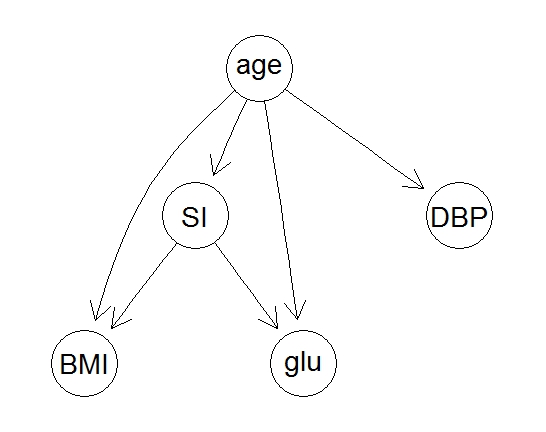}\label{fig:f1}}
  \qquad
  \subfloat[dHSIC]{\includegraphics[width=0.4\textwidth]{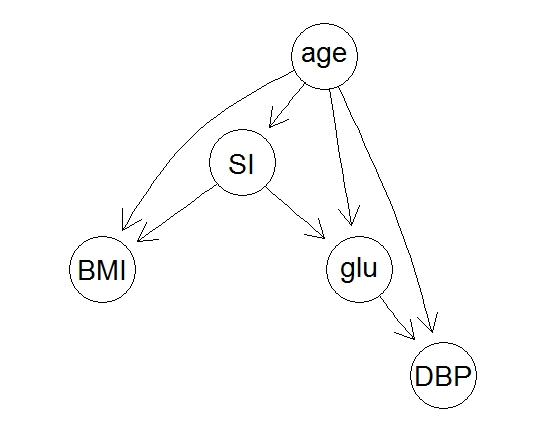}\label{fig:f2}}
  \caption{The DAG models corresponding to the largest p-values from the five tests.}\label{fig1}
\end{figure}

Figure \ref{fig:f1} shows the model with the maximum p-value among all the 48 candidate DAG models, when the test for joint independence of the residuals is conducted based on $\widetilde{JdCov^2}$, $\widetilde{JdCov^2_S}$ and $\widetilde{JdCov^2_R}$ and $T_{MT}$. This graphical structure goes in tune with the biological evidences of causal relationships among these five variables. Figure \ref{fig:f2} stands for the model with the maximum p-value when the test is based on dHSIC. Its only difference with Figure \ref{fig:f1} is that, it has an additional edge from glu to DBP, indicating a causal effect of Plasma Glucose Concentration on Diastolic Blood Pressure. We are unsure of any biological evidence that supports such a causal relationship in reality.

\begin{remark}\label{DAG_real_choose_c}
\rm
In view of Remark \ref{choice_c}, it might be intriguing to take into account the heuristic data-driven way of determining $c$ in the above example, instead of setting $c$ at a default value of $1$. Our findings indicate that choosing $c$ in the data-driven way leads to a slightly different result. The tests based on dHSIC and $\widetilde{JdCov^2_S}$ select the DAG model shown in Figure \ref{fig:f2} (considering the maximum p-value among all the $48$ candidate DAG models), whereas Figure \ref{fig:f1} is the DAG model selected when the test is based on $\widetilde{JdCov^2}$, $\widetilde{JdCov^2_R}$ and $T_{MT}$. The proposed tests (based on $\widetilde{JdCov^2}$ and $\widetilde{JdCov^2_R}$) still perform well.

\end{remark}

\subsection{A simulation study}\label{DAG sim}
We conduct a simulation study based on our findings in the previous real data example. To save the computational cost, we focus our attention on three of the five variables, viz. Age, glu and DBP. In the correct causal structure among these three variables, there are directed edges from Age to glu and
Age to DBP. We consider the additive structural equation models
\begin{align}\label{additive2}
X_j = \displaystyle \sum_{k \in \text{Par}(j)} \hat{f}_{j,k}(X_k) + e_j \, , \; j=1,2,3,
\end{align}
where $X_1,X_2,X_3$ correspond to Age, glu and DBP (after centering and scaling) respectively, and $\hat{f}_{j,k}$ denotes the estimated function from the real data. Note that $X_1$ is the only variable without any parent. In Section \ref{6.2}, we get from our numerical studies that the standard deviation of $X_1$ is 1.001, and the standard deviations of the residuals when $X_2$ and $X_3$ are regressed on $X_1$ (according to the structural equation models in ($\ref{additive}$), are 0.918 and 0.95, respectively. In this simulation study, we simulate $X_1$ from a zero mean Gaussian distribution with standard deviation 1. For $X_2$ and $X_3$, we simulate the noise variables from zero mean Gaussian distributions with standard deviations 0.918 and 0.95, respectively. The same $n=392$ is considered for the number of generated observations, and based on this simulated dataset we perform the model diagnostic checking for $27$ candidate DAG models. The number of bootstrap replications is set to be $B=100$ (to save the computational cost). This procedure is repeated 100 times to note how many times out of 100 that the five tests select the correct model, based on the largest p-value. The results in Table \ref{tb2} indicate that the proposed tests with $c=1$ and the dHSIC-based test outperform $T_{MT}$. \\

\begin{table}[!h]\footnotesize
\centering
\caption{The number of times (out of 100) that the true model is being selected.}
\label{tb2}
\begin{tabular}{ccccc}
\toprule
$\widetilde{JdCov^2}$ & $\widetilde{JdCov^2_S}$ & $\widetilde{JdCov^2_R}$ & dHSIC & $T_{MT}$ \\
\hline
45 & 61 & 54 & 52 & 32\\
\toprule
\end{tabular}
\end{table}

\vspace{0.3cm}

\begin{remark}\label{res boot}
\rm
A natural question to raise is why do we bootstrap the residuals and not test for the joint independence of the estimated residuals directly, to check for the goodness of fit of the DAG model. From the idea in Sen and Sen (2014), it appears that the joint distance covariance of the estimated residuals might have a limiting distribution different from the one stated in Proposition \ref{prop6}. We leave the formulation of a rigorous theory in support of that for future research. We present below the models selected most frequently (out of 100 times) by the different test statistics if we repeat the simulation study done above in Section \ref{DAG sim} without using residual bootstrap to re-estimate $f_{j,k}$. We immediately see that joint independence tests of the estimated residuals based on all of the five statistics we consider, select a DAG model that is meaningless and far away from the correct one.

\begin{figure}[h]
  \centering
  \subfloat[$\widetilde{JdCov^2}$, $\widetilde{JdCov^2_S}$, $\widetilde{JdCov^2_R}$, dHSIC]{\includegraphics[width=0.32\textwidth]{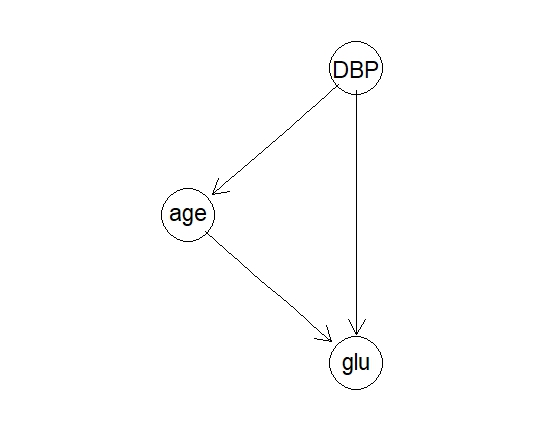}\label{fig:f3}}
  \;
  \subfloat[$T_{MT}$]{\includegraphics[width=0.32\textwidth]{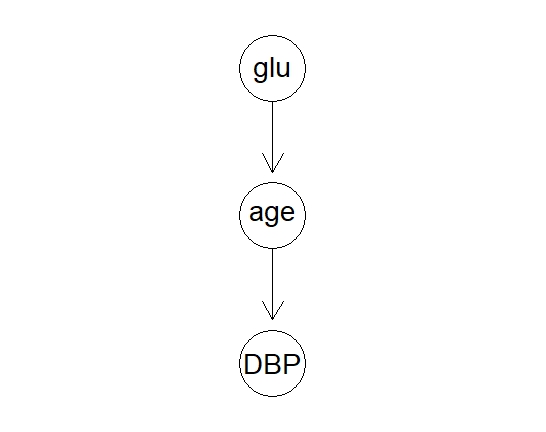}\label{fig:f4}}
  \;
  \subfloat[Correct model]{\includegraphics[width=0.32\textwidth]{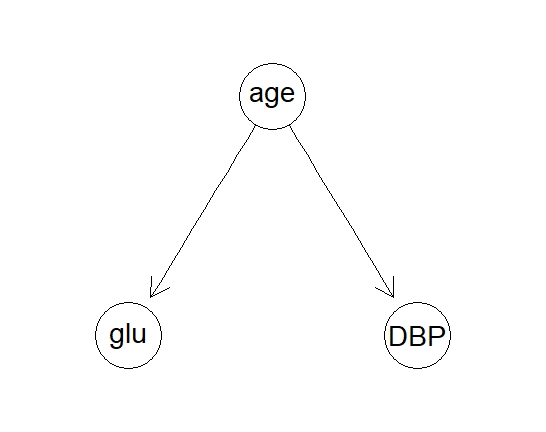}\label{fig:f5}}
  \caption{The DAG models selected (most frequently out of 100 times) by the five tests, without doing residual bootstrap to re-estimate $f_{j,k}$.}
  \label{fig2}
\end{figure}

\end{remark}

\begin{remark}\label{res boot_c}
\rm
In view of Remark \ref{choice_c}, it might be intriguing to take into account the heuristic data-driven way of choosing $c$ in the simulation study in Section \ref{DAG sim}, instead of setting $c$ at a default value of $1$. Our findings indicate that our proposed tests and the dHSIC-based test still outperform $T_{MT}$. In the context of Remark \ref{res boot}, if we repeat the simulation study done in Section \ref{DAG sim} (choosing $c$ in the heuristic way), we still reach the same conclusion presented in Remark \ref{res boot}.

\end{remark}

\section{Discussions}\label{sec:comp}

Huo and Sz\'ekely (2016) proposed an $O(n \log n)$ algorithm to compute dCov of univariate random variables. In a more recent work, Huang and Huo (2017) introduced a fast method for multivariate cases which is based on random projection and has computational complexity $O(n K \log n)$, where $K$ is the number of random projections. One of the possible directions for future research is to come up with a fast algorithm to compute JdCov. When $p_i=1$, we can indeed use the method in Huo and Sz\'ekely (2016) to compute JdCov. But their method may be inefficient when $d$ is large and it is not applicable to the case where $p_i>1$. Another direction is, to introduce the notion of Conditional JdCov in light of Wang et al. (2015), to test if the variables $(X_1,\dots,X_d)$ are jointly independent given another variable $Z$.

\vspace{0.3cm}

\section{Acknowledgments}

We acknowledge partial funding from the NSF grant DMS-1607320. We are grateful to the Editor and two very careful reviewers for their immensely helpful comments and suggestions which greatly improved the paper. It would also be our pleasure to thank Niklas Pfister for sharing with us the R codes used in Pfister et al. (2018).

\section{Supplementary materials}
The supplementary materials contain some additional proofs of the main results in the paper and tabulated numerical results from Section \ref{sec:num}.

\begin{proof}[Proof of Lemma \ref{lem1}]
By Lemma 1 of Sz\'{e}kely et al. (2007), we have
\begin{align*}
RHS=&\int_{\mathbb{R}^p}\bigg\{\E e^{\I\langle
t,X_i-X'_i\rangle}+e^{\I \langle t,x-x'\rangle}-\E e^{\langle
t,x-X'_i\rangle}-\E e^{\I \langle
t,X_i-x'\rangle}\bigg\}w_{p_i}(t)dt
\\=&\E\int_{\mathbb{R}^p}\bigg\{\cos(\langle t,X_i-X'_i\rangle)-1+\cos(\langle t,x-x'\rangle)-1+1-\cos(\langle t,x-X'_i\rangle)
\\&+1-\cos(\langle t,X_i-x'\rangle)\bigg\}w_{p_i}(t)dt+\I\int_{\mathbb{R}}\E\bigg\{\sin(\langle t,X_i-X'_i\rangle)+\sin(\langle t,x-x'\rangle)
\\&-\sin(\langle t,x-X'_i\rangle)-\sin(\langle t,X_i-x'\rangle)\bigg\}w_{p_i}(t)dt
\\=&\E|x-X'_i|+\E|X_i-x'|-|x-x'|-\E|X_i-X'_i|\, =\, U_i(x,x').
\end{align*}
Here we have used the fact that $\int_{\mathbb{R}}\{\sin(\langle
t,X-X'\rangle)+\sin(\langle t,x-x'\rangle)-\sin(\langle
t,x-X'\rangle)-\sin(\langle t,X-x'\rangle)\}w_{p_i}(t)dt=0.$
\end{proof}

\begin{proof}[Proof of Proposition \ref{prop0}]
To show (1), notice that for $a_i$, $c_i$ and orthogonal
transformations $A_i\in\mathbb{R}^{p_i\times p_i}$,
$$\E\prod_{i\in S}U_i(a_i+c_iA_iX_i,a_i+c_iA_iX_i')=\prod_{i\in S}|c_i|\,\E\prod_{i\in S}U_i(X_i,X_i'),$$
where $S\subset \{1,2,\dots,d\}$. The conclusion follows directly.
\end{proof}

\begin{proof}[Proof of Proposition \ref{prop:normality}]
The proof is essentially similar to the proof of Lemma 1.2 in the supplementary material of Yao et al. (2018). It is easy to verify that $$ E \prod_{i=1}^d U_i(X_i, X_i') = C \int_{\mathbb{R}^d} |A|^2 \frac{dt_1}{t_1^2}\dots \frac{dt_d}{t_d^2} \; ,$$ where $C$ is some constant and
\begin{align}\label{splitup}
\begin{split}
A = e^{-\frac{t_1^2 + \dots + t_d^2}{2}}\; &\Big[ e^{-\rho t_1 t_2} +e^{-\rho t_1 t_3} + \dots \;\; \binom{d}{2}\; \textrm{similar terms} \\
& - e^{-\rho t_1 t_2 - -\rho t_1 t_3 -\rho t_2 t_3} - e^{-\rho t_1 t_2 -\rho t_1 t_4 -\rho t_2 t_4} - \dots \;\; \binom{d}{3}\; \textrm{similar terms}\\
& + e^{-\rho t_1 t_2 - -\rho t_1 t_3 -\rho t_1 t_4 -\rho t_2 t_3 -\rho t_2 t_4 -\rho t_3 t_4} + \dots \;\; \binom{d}{4}\; \textrm{similar terms}\\
& + \;\;\dots \;\; - (d-1) \;\; \Big] \; .
\end{split}
\end{align}
For example, if $d \geq 4$ and we use the Taylor's expansion $e^x = 1 + x + \frac{x^2}{2!} + \displaystyle \sum_{l=3}^{\infty}\frac{x^l}{l!}$, then keeping in mind the multinomial theorem $$(a_1 + \dots + a_q)^2 = \displaystyle \sum_{\substack{l_1, \dots, l_q=0 \\ l_1 + \dots + l_q =2}}^2 \frac{2!}{l_1! \dots l_q!}\;\; ,\;\; q\geq 2 \; ,$$ it is easy to check that the leading terms and their coefficients (upto some constants) are

\begin{table}[h!]
\centering
\begin{tabular}{c | c}
 Leading terms & Coefficients (upto some constants)  \\ [0.5ex]
 \hline
 $t_i t_j$ & $1 - \binom{d-2}{1} + \binom{d-2}{2} - \dots$ \; \; ($=0$ for $d>2$)  \\
 $t_i^2 t_j^2$ & $1 - \binom{d-2}{1} + \binom{d-2}{2} - \dots$ \; \; ($=0$ for $d>2$) \\
 $t_i^2 t_j t_k$ & $-1 + \binom{d-3}{1} - \binom{d-3}{2} + \dots$ \; \; ($=0$ for $d>3$)  \\
 $t_i t_j t_k t_l$ & \qquad \qquad $1 - \binom{d-4}{1} + \binom{d-4}{2} - \dots$ \;\; ($=1$ if $d=4$, $=0$ for $d>4$)  \\

 \hline
\end{tabular}\; .
\end{table}

To get a non-trivial upper bound for $ E \displaystyle\prod_{i=1}^d U_i(X_i, X_i')$, we need to consider the Taylor's expansion $e^x = 1 + x + \frac{x^2}{2!} + \dots + \frac{x^k}{k!}  + \displaystyle \sum_{l=k+1}^{\infty} \frac{x^l}{l!}$, when $d=2k-1$ or $d=2k$, $k\geq2$, and the only leading term in the Taylor's expansion of (\ref{splitup}) that would lead to a term with non-vanishing coefficient, is $\frac{x^k}{k!}$. To see this, note that when $d=4$, i.e., $k=2$, it is shown in Lemma 1.2 in the supplementary material of Yao et al. (2018) that the only non-vanishing term is $t_1 t_2 t_3 t_4$ (upto some constants). Likewise for $d=5$ and $6$, the only non-vanishing leading terms (upto some constants) are :

\begin{table}[!h]
\centering
\begin{tabular}{c | c | c}
 $d$ & $k$ & The only non-vanishing term (upto some constants)  \\ [0.5ex]
 \hline
 5 & 3 & $(t_i t_j)^1 (t_i t_a)^1 (t_l t_m)^1 = t_i^2 t_j t_a t_l t_m$  \\
 6 & 3 & $(t_i t_j)^1 (t_a t_l)^1 (t_m t_n)^1 = t_i t_j t_a t_l t_m t_n$ \\
 \hline
\end{tabular}\; ,
$i \neq j \neq a \neq l \neq m \neq n$ .
\end{table}

\newpage

In general when $d=2k-1$, for $k\geq 2$, the only non-vanishing term (upto some constants) is $t_{i_1}^2 t_{i_2}\dots t_{i_d}$, where $(i_1, \dots, i_d)$ is any permutation of $(1,2,\dots,d)$. Suppose $P_d$ denotes the set of all possible permutations of $(1,2,\dots,d)$. Then
\begin{align*}
E \displaystyle\prod_{i=1}^d U_i(X_i, X_i')\; &=\; c_0 \int_{\mathbb{R}^d}|e^{-\frac{t_1^2 + \dots + t_d^2}{2}} (\,c_1 \rho^k \displaystyle \sum_{\substack{(i_1, \dots, i_d) \in P_d}} t^2_{i_1} t_{i_2} \dots t_{i_d} + R\,)\,|^2 \, \frac{dt_1}{t_1^2}\dots \frac{dt_d}{t_d^2} \\
&= A_0 + A_1 + A_2 + A_3 \; ,
\end{align*}

where $$A_0\; = \;\tilde{c}_0 \,\rho^{2k}\, \displaystyle \sum_{\substack{(i_1, \dots, i_d) \in P_d}} \, \int_{\mathbb{R}^d} e^{-(t_1^2 + \dots + t_d^2)}\; t_{i_1}^4 t_{i_2}^2 \dots t_{i_d}^2 \;\frac{dt_1}{t_1^2}\dots \frac{dt_d}{t_d^2} \;,$$ $$A_1\; = \;\tilde{c}_1 \,|\rho|^{k}\, \displaystyle \sum_{\substack{(i_1, \dots, i_d) \in P_d}} \, \int_{\mathbb{R}^d} e^{-(t_1^2 + \dots + t_d^2)}\; t_{i_1}^2 t_{i_2} \dots t_{i_d} \times R \;\frac{dt_1}{t_1^2}\dots \frac{dt_d}{t_d^2}\; ,$$ $$A_2\; = \;\tilde{c}_2 \,\rho^{2k}\, \displaystyle \sum_{\substack{(i_1, \dots, i_d) \in P_d}} \, \int_{\mathbb{R}^d} e^{-(t_1^2 + \dots + t_d^2)}\; t_{i_1}^3 t_{i_2}^3 t_{i_3}^2 \dots t_{i_d}^2 \;\frac{dt_1}{t_1^2}\dots \frac{dt_d}{t_d^2}\; ,$$ and $$ A_3\; = \;\tilde{c}_3  \int_{\mathbb{R}^d} e^{-(t_1^2 + \dots + t_d^2)}\; \times R^2 \;\frac{dt_1}{t_1^2}\dots \frac{dt_d}{t_d^2}\; ,$$ $c_0, c_1, \tilde{c}_0, \tilde{c}_1$, $\tilde{c}_2$ and $\tilde{c}_3$ being some constants and $R$ being the remainder term from the Taylor's expansion. Following the similar arguments of Yao et al. (2018), it can be shown that $$ A_0 \;= \;O(|\rho|^{2k})\; , \; A_1\;= \;O(|\rho|^{2k+1})\;, \; A_2 \;= \;O(|\rho|^{2k})\; \textrm{and} \; A_3\;= \;O(|\rho|^{2k+2})\;.$$ Thus for $d=2k-1, k\geq 2$, $$E \displaystyle\prod_{i=1}^d U_i(X_i, X_i')\;= \; O(|\rho|^{2k})\;.$$

And when $d=2k$, for $k\geq 2$, the only non-vanishing term (upto some constants) is $t_1 t_2 \dots t_d$. Consequently
\begin{align*}
E \displaystyle\prod_{i=1}^d U_i(X_i, X_i')\; &=\; c_0' \int_{\mathbb{R}^d}|e^{-\frac{t_1^2 + \dots + t_d^2}{2}} (\,c_1' \rho^k t_1 t_2 \dots t_d + R'\,)\,|^2 \, \frac{dt_1}{t_1^2}\dots \frac{dt_d}{t_d^2} \\
&= A_0' + A_1' + A_2' \; ,
\end{align*}

where $$A_0'\; = \;\tilde{c}_0' \,\rho^{2k} \int_{\mathbb{R}^d} e^{-(t_1^2 + \dots + t_d^2)}\; t_1^2 t_2^2 \dots t_d^2 \;\frac{dt_1}{t_1^2}\dots \frac{dt_d}{t_d^2} \;,$$ $$A_1'\; = \;\tilde{c}_1' \,|\rho|^{k} \int_{\mathbb{R}^d} e^{-(t_1^2 + \dots + t_d^2)}\; t_1 t_2 \dots t_d \times R' \;\frac{dt_1}{t_1^2}\dots \frac{dt_d}{t_d^2}\; ,$$ and $$ A_2'\; = \;\tilde{c}_2'  \int_{\mathbb{R}^d} e^{-(t_1^2 + \dots + t_d^2)}\; \times R'^2 \;\frac{dt_1}{t_1^2}\dots \frac{dt_d}{t_d^2}\; ,$$ $c_0', c_1', \tilde{c}_0', \tilde{c}_1'$ and $\tilde{c}_2'$ being some constants and $R'$ being the remainder term from the Taylor's expansion. Again following the similar arguments of Yao et al. (2018), it can be shown that $$ A_0' \;= \;O(|\rho|^{2k})\; , \; A_1'\;= \;O(|\rho|^{2k+1})\; \textrm{and} \; A_2'\;= \;O(|\rho|^{2k+2})\;.$$ Thus for $d=2k, k\geq 2$, $$E \displaystyle\prod_{i=1}^d U_i(X_i, X_i')\;= \; O(|\rho|^{2k})\;,$$ which completes the proof.

\end{proof}

\begin{proposition}\label{prop01}

\begin{enumerate}[(1)]

\item $dCov^2(X_1,\dots,X_d)\leq \E[\prod^{d}_{j=1}\min\{a_j(X_j),a_j(X_j')\}]$ with $a_j(x)=\max\{\E|X_j-X_j'|,\,|\E|X_j-X_j'|-2\E|x-X_j||\}.$
For any partition $S_1\cup S_2=\{1,2,\dots,d\}$ and $S_1\cap S_2=\emptyset$, we have
$dCov^2(X_1,\dots,X_d)\leq \E[\prod_{i\in
S_1}a_j(X_j)]\, \E[\prod_{i\in S_2}a_j(X_j)].$
\item $dCov^2(X_1,\dots,X_d) \leq
\prod^{d}_{i=1}\left\{\E[|U_i(X_i,X_i')|^d]\right\}^{1/d}.$
In particular, when $d$ is even, $dCov^2(X_1,\dots,X_d)\leq
\prod_{i=1}^{d}dCov^2(\underbrace{X_i,\dots,X_i}_{d})^{1/d}$.
\item Denote by $\mu_j$ the uniform probability measure on the unit sphere $\mathcal{S}^{p_j-1}$. Then
\begin{align*}
&dCov^2(X_1,\dots,X_d)=C\int_{\prod_{j=1}^{d}\mathcal{S}^{p_j-1}}dCov^2(\langle u_1,X_1\rangle,\dots,\langle u_d,X_d\rangle)d\mu_1(u_1)\cdots d\mu_d(u_d),
\end{align*}
and \begin{align*}
&JdCov^2(X_1,\dots,X_d;c)
\\=&C'\int_{\prod_{j=1}^{d}\mathcal{S}^{p_j-1}}JdCov^2(\langle u_1,X_1\rangle,\dots,\langle u_d,X_d\rangle;c)d\mu_1(u_1)\cdots d\mu_d(u_d),
\end{align*}
for some positive constants $C$ and $C'$.
\end{enumerate}
\end{proposition}

\begin{proof}[Proof of Proposition \ref{prop01}]
To prove (1),
we have by the triangle inequality
\begin{align*}
\left|\E[|X_j-x'|]-|x-x'|\right|\leq \E[|x-X_j'|].
\end{align*}
 Thus we have $|U_j(x,x')|\leq
\min\{a_j(x),a_j(x')\},$ which implies that
$$\E\left[\prod^{d}_{j=1}U_j(X_j,X'_j)\right]\leq \E\left[\prod^{d}_{j=1}\min\{a_j(X_j),a_j(X_j')\}\right].$$
For any partition $S_1\cup S_2=\{1,2,\dots,d\}$ and $S_1\cap S_2=\emptyset$, using the independence between $X_j$ and $X_j'$, we get
$$dCov^2(X_1,X_2,\dots,X_d)\leq \E\left[\prod_{i\in S_1}a_j(X_j)\right]\E\left[\prod_{i\in S_2}a_j(X_j)\right].$$
(2) follows from the H\"{o}lder's inequality directly. Finally, by the change of variables: $t_1=r_i u_i$ where $r_i \in (0,+\infty)$ and $u_i \in \mathcal{S}^{p_i-1}$, we have
\begin{align*}
&dCov^2(X_1,X_2,\dots,X_d)
\\=&\int_{\mathbb{R}^{p_0}}\left|\E\left[\prod_{i=1}^{d}(f_i(t_i)-e^{\I \langle t_i,X_i\rangle})\right]\right|^2dw
\\=&C_1\int_{\mathcal{S}_+^{p_1}}\cdots \int_{\mathcal{S}_+^{p_d}}\int_{-\infty}^{+\infty}\cdots \int_{-\infty}^{+\infty}\left|\E\left[\prod_{i=1}^{d}(\E e^{\I r_i\langle u_i,X_i\rangle}-e^{\I r_i\langle u_i,X_i\rangle})\right]\right|^2\prod^{d}_{i=1}d\mu_i(u_i)dr_i
\\=&C_2\int_{\mathcal{S}_+^{p_1}}\cdots \int_{\mathcal{S}_+^{p_d}}JdCov^2(\langle u_1,X_1\rangle,\dots,\langle u_d,X_d\rangle;c)d\mu_1(u_1)\cdots d\mu_d(u_d)
\\=&C_3\int_{\mathcal{S}^{p_1}}\cdots \int_{\mathcal{S}^{p_d}}JdCov^2(\langle u_1,X_1\rangle,\dots,\langle u_d,X_d\rangle;c)d\mu_1(u_1)\cdots d\mu_d(u_d),
\end{align*}
where $C_1,C_2,C_3$ are some positive constants.
\end{proof}

Property (1) gives an upper bound for
$dCov^2(X_1,X_2,\dots,X_d)$, which is motivated by Lemma 2.1
of Lyons (2013), whereas an alternative upper bound is given in
Property (2) which follows directly from the H\"{o}lder's
inequality. Property (3) allows us to represent $dCov$ of random
vectors of any dimensions as an integral of $dCov$ of univariate
random variables, which are the projections of the aforementioned
random vectors.

\begin{proof}[Proof of Proposition \ref{prop1}]
The ``if'' part is trivial. To prove the ``only if'' part, we
proceed using induction. Clearly this is true if $d=2$. Suppose the
result holds for $d=m$. Note that $dCov^2(X_1,X_2,\dots,X_{m+1})=0$
implies that $\E\left[\prod_{i=1}^{m+1}(f_i(t_i)-e^{\I \langle t_i
X_i\rangle})\right]=0$ almost everywhere. Thus we can write the higher
order effect $f_{12\cdots (m+1)}(t_1,\dots,t_{m+1})-\prod^{m+1}_{i=1} f_i(t_i)$ as
a linear combination of the lower order effects. By the assumption that
$(X_{i_1},\dots,X_{i_m})$ are mutually independent for any
$m$-tuples in $I^d_m$ with $m<d$, we know $f_{12\cdots
(m+1)}(t_1,\dots,t_{m+1})-\prod^{m+1}_{i=1} f_i(t_i)=0.$
\end{proof}

\begin{proof}[Proof of Proposition \ref{prop2}]
Notice that
\begin{align*}
&\prod^{d}_{i=1}\left( U_i(X_i,X'_i)+c\right)
\\=&c^d+c^{d-1}\sum_{i=1}^{d}U_i(X_i,X_i')+c^{d-2}\sum_{(i_1,i_2)\in
I^d_2}U_{i_1}(X_{i_1},X_{i_1}')U_{i_2}(X_{i_2},X_{i_2}')
\\&+\cdots+\prod^{d}_{i=1}U_i(X_i,X'_i).
\end{align*}
The conclusion follows from the fact that
$\E[U_i(X_i,X_i')]=0$, equation (\ref{eqn2}) and the definition of JdCov.
\end{proof}

\begin{proof}[Proof of Proposition \ref{prop-cha}]
We only prove the ``if'' part. If $dcf(t_1,\dots,t_d)$ can be
factored, $U(X_i,X_i')$ are independent. Therefore, it is
easy to see that $JdCov^2(X_1,\dots,X_d;c)=0$,
which implies that $\{X_1,\dots,X_d\}$ are mutually independent by
Proposition \ref{prop1}.
\end{proof}

\begin{proof}[Proof of Lemma \ref{lemma2}]
The RHS of
(\ref{eq-lemma2}) in the main paper is equal to
\begin{align*}
&\frac{1}{n^2}\sum^{n}_{k,l=1}\prod_{i=1}^{m}\int_{\mathbb{R}}(\hat{f}_i(t_i)-e^{\I
\langle t_i,X_{ki}\rangle})(\hat{f}_i(-t_i)-e^{-\I \langle t_i,
X_{li}\rangle})w_{p_i}(t_i)dt_i.
\end{align*}
Thus it is enough to prove that
\begin{align*}
\int (\hat{f}_i (t_i) - e^{\imath \langle t_i , X_{ki}\rangle})(\hat{f}_i (-t_i) - e^{-\imath \langle t_i , X_{li}\rangle}) \, w_{p_i}(t_i) dt_i \, &= \, \widehat{U}_{i}(k,l).
\end{align*}
To this end, we note that
\begin{align*}
&(\hat{f}_i (t_i) - e^{\imath \langle t_i, X_{ki}\rangle})\, (\hat{f}_i (-t_i) - e^{-\imath \langle t_i , X_{li}\rangle})
\\ = \; & \hat{f}_i (t_i)\hat{f}_i (-t_i) \,-\, e^{\imath \langle t_i , X_{ki}\rangle}\hat{f}_i(-t_i)\, -\, \hat{f}_i(t_i) e^{-\imath \langle t_i , X_{li}\rangle}\, +\, e^{\imath \langle t_i ,(X_{ki} - X_{li})\rangle}
\\ = \;& \frac{1}{n^2} \displaystyle \sum_{k,l=1}^n e^{\imath \langle t_i , (X_{ki} - X_{li})\rangle}\, -\, \frac{1}{n} \displaystyle \sum_{l=1}^n e^{\imath \langle t_i ,(X_{ki} - X_{li})\rangle}\, -\, \frac{1}{n} \displaystyle \sum_{k=1}^n e^{\imath \langle t_i ,(X_{ki} - X_{li})\rangle}\, + \,e^{\imath \langle t_i ,(X_{ki} - X_{li})\rangle}
\\ = \;& \frac{1}{n} \displaystyle \sum_{l=1}^n ( 1 - e^{\imath \langle t_i ,(X_{ki} - X_{li})\rangle} )\, + \, \frac{1}{n} \displaystyle \sum_{k=1}^n ( 1 - e^{\imath \langle t_i ,(X_{ki} - X_{li})\rangle} )- ( 1 - e^{\imath \langle t_i ,(X_{ki} - X_{li})\rangle} )
\\& -\frac{1}{n^2} \displaystyle \sum_{k,l=1}^n ( 1 - e^{\imath \langle t_i ,(X_{ki} - X_{li})\rangle} )  \; .
\end{align*}
Using (2.11) of Sz\'{e}kely et al. (2007), we obtain
\begin{align*}
&\int (\hat{f}_i (t_i) - e^{\imath \langle t_i , X_{ki}\rangle}) (\overline{\hat{f}_i (t_i)} - e^{-\imath \langle t_i , X_{li}\rangle}) \, w_{p_i}(t_i) dt_i
\\ = &\; \frac{1}{n} \displaystyle \sum_{l=1}^n |X_{ki} - X_{li}|_{p_i} \, + \, \frac{1}{n} \displaystyle \sum_{k=1}^n |X_{ki} - X_{li}|_{p_i} \, - \, |X_{ki} - X_{li}|_{p_i} \, - \, \frac{1}{n^2} \displaystyle \sum_{k,l=1}^n |X_{ki} - X_{li}|_{p_i}
\\ = &\; \widehat{U}_{i} (k,l) \; .
\end{align*}
Finally, (\ref{eq2-lemma2}) in the paper follows from (\ref{eq-lemma2}) in the paper and the
definition of $\widehat{JdCov^2}$.
\end{proof}


\begin{proof}[Proof of Proposition \ref{prop5}]
Define
\begin{align*}
\xi (t_1, t_2,\dots,t_d) = &\, \E\left[\prod_{i=1}^d (f_i (t_i) -
e^{\imath \langle t_i , X_i \rangle})\right],\quad \xi_n (t_1,
t_2,\dots,t_d)=\frac{1}{n} \displaystyle \sum_{j=1}^n \prod_{i=1}^d
(\hat{f_i} (t_i) - e^{\imath \langle t_i , X_{ji}\rangle}),
\end{align*}
and note that
\begin{align*}
&dCov^2(X_1, X_2,\cdots, X_d) = \int \left| \xi (t_1, t_2,\dots,t_d)\right|^2 dw,\\
&\widehat{dCov^2}(X_1, X_2,\cdots, X_d) = \int \left| \xi_n (t_1, t_2,\dots,t_d)\right|^2 dw.
\end{align*}
Direct calculation shows that
\begin{align*}
\xi (t_1, t_2,\dots,t_d) = &\prod^{d}_{j=1}f_j-d\prod^{d}_{j=1}f_j+\left(f_{12}\prod_{j\neq 1,2}f_j+f_{13}\prod_{j\neq 1,3}f_j+\cdots\right)
\\ &-\left(f_{123}\prod_{j\neq 1,2,3} f_j + f_{124}\prod_{j\neq 1,2,4}f_j+\cdots\right)+\cdots+(-1)^df_{12,\dots,d},
\end{align*}
and $ \xi_n (t_1, t_2,\dots,t_d)$ has the same expression by
replacing the characteristic functions by their empirical
counterparts in $\xi (t_1, t_2,\dots,t_d)$.
Then by the strong law of large numbers, we have for any fixed $(t_1,t_2,\dots,t_d)$,
\begin{align*}
\xi_n (t_1, t_2,\dots,t_d)\overset{a.s}{\longrightarrow}\xi (t_1,t_2,\dots,t_d).
\end{align*}
For complex numbers $x_1, x_2,\dots,x_n$ with $n\geq2$, the CR inequality says that for any $r>1$
\begin{equation}\label{eq8}
\left|\sum_{i=1}^n x_i\right|^r \leq n^{r-1}\sum_{i=1}^n |x_i|^r.
\end{equation}
Using (\ref{eq8}), we get
\begin{align*}
|\xi_n (t_1, t_2,\dots, t_d)|^2 = &\; \left|\,\frac{1}{n}
\displaystyle \sum_{j=1}^n \prod_{i=1}^d (\hat{f_i} (t_i) -
e^{\imath \langle t_i , X_{ji}\rangle})\,\right|^2 \leq
\frac{1}{n^2}\, n^{2-1} \displaystyle \sum_{j=1}^n \prod_{i=1}^d
\left|\hat{f_i} (t_i) - e^{\imath \langle t_i ,
X_{ji}\rangle}\right|^2
\\ = & \; \frac{1}{n} \displaystyle \sum_{j=1}^n \prod_{i=1}^d 4 =4^d.
\end{align*}
For any $\delta > 0$, define $D(\delta) \, = \, \{(t_1, t_2,\dots,t_d) \, : \, \delta \leq |t_i|_{p_i} \leq 1/\delta \, ,\; i=1,2,\dots,d\}.$ Notice that
\begin{equation*}
\begin{split}
\widehat{dCov^2}(X_1, X_2,\dots,X_d) = & \; \int_{D(\delta)} | \,\xi_n (t_1, t_2,\dots, t_d)\,|^2 \, dw \, + \, \int_{D^c(\delta)} | \,\xi_n (t_1, t_2,\dots,t_d)\,|^2 dw
\\ = & \;\; D_{n,\delta}^{(1)} \, + \, D_{n,\delta}^{(2)} \quad (\text{say}),
\end{split}
\end{equation*}
where $D_{n,\delta}^{(1)}\leq \int_{D(\delta)} 4^d < \infty$. Using
the Dominated Convergence Theorem (DCT), we have as $n \to \infty$ ,
\[ D_{n,\delta}^{(1)} \, \overset{a.s}{\longrightarrow}\,
\int_{D(\delta)} | \,\xi (t_1, t_2, ..\,, t_d)\,|^2 \, dw \, = \,
D_{\delta}^{(1)}\quad \text{(say)}.\] So, almost surely
\begin{equation*}
\begin{split}
\lim_{\delta \to 0} \lim_{n \to \infty} \, D_{n,\delta}^{(1)} \, = & \, \lim_{\delta \to 0}\, D_{\delta}^{(1)}=\int | \,\xi (t_1, t_2, ..\,, t_d)\,|^2 \, dw =dCov^2(X_1, X_2, \, .. \, , X_d).
\end{split}
\end{equation*}
The proof will be complete if we can show almost surely
\[ \lim_{\delta \to 0} \lim_{n \to \infty} \, D_{n,\delta}^{(2)} \, = \, 0. \]

To this end, write $D^c(\delta) = \bigcup_{i=1}^d (A_i^{1} \cup
A_i^{2})$, where $A_i^{1} = \{|t_i|_{p_i} < \delta\}$ and $A_i^{2} =
\{|t_i|_{p_i} > \frac{1}{\delta}\}$ for $i=1,2,\dots,d.$ Then we
have
$$ D_{n,\delta}^{(2)}= \int_{D^c(\delta)} | \,\xi_n (t_1, t_2,\dots,t_d)\,|^2 \, dw \, \leq \, \sum_{\substack{
   i=1,2,\dots,d \\
   k=1,2
  }} \int_{A_i^{k}}|\xi_n (t_1, t_2,\dots,t_d)|^2dw.$$
Define $u_j^i=e^{\imath \langle t_i , X_{ji}\rangle} - f_i (t_i)$
for $1\leq j\leq n$ and $1\leq i\leq d.$ Following the proof of
Theorem 2 of Sz\'{e}kely \textit{et al.} (2007), we have for $i=1,
2,\dots,d$,
\begin{align}
&\int_{\mathbb{R}^{p_i}} \frac{|u_j^i|^2}{c_{p_i} |t_i|_{p_i}^{1+p_i}}\, dt_i \; \leq \; 2 \,(\, |X_{ji}| + \E|X_i|\,) ,  \label{as-1}\\
&\int_{|t_i|_{p_i} < \delta} \frac{|u_j^i|^2}{c_{p_i} |t_i|_{p_i}^{1+p_i}}\, dt_i \; \leq \; 2 \,\E[|X_{ji} - X_i||X_{ji}] \, G(\,|X_{ji} - X_i|\delta \,) , \label{as-2}\\
&\int_{|t_i|_{p_i} > 1/\delta} \frac{|u_j^i|^2}{c_{p_i} |t_i|_{p_i}^{1+p_i}}\, dt_i \; \leq \; 4 \delta, \label{as-3}
\end{align}
where
\begin{equation*}
G(y) = \int_{|z| < y} \frac{1 - \cos z_1 }{|z|^{1+p}}dz,
\end{equation*}
which satisfies that $ G(y) \leq c_p$ and $\displaystyle \lim_{y \to 0} G(y) = 0$.
Notice that
\[\xi_n (t_1, t_2,\dots,t_d)=\frac{1}{n} \displaystyle \sum_{j=1}^n \prod_{i=1}^d \left(\frac{1}{n}\displaystyle \sum_{k=1}^n u_k^i - u_j^i\right).\]
Some algebra yields that
\begin{align*}
&\xi_n (t_1, t_2,\dots,t_d)
\\=& \prod^{d}_{i=1}\left(\frac{1}{n} \displaystyle\sum_{k=1}^n u_k^i\right)-d\prod^{d}_{i=1}\left(\frac{1}{n} \displaystyle\sum_{k=1}^n u_k^i\right)
\\ & + \left\{
\left(\frac{1}{n} \displaystyle\sum_{k=1}^n u_k^1u_k^2\right)\prod_{i\neq 1,2}\left(\frac{1}{n} \displaystyle\sum_{k=1}^n u_k^i\right)+
\left(\frac{1}{n} \displaystyle\sum_{k=1}^n u_k^1u_k^3\right)\prod_{i\neq 1,3}\left(\frac{1}{n} \displaystyle\sum_{k=1}^n u_k^i\right)+\cdots\right\}
\\ & + \left\{
\left(\frac{1}{n} \displaystyle\sum_{k=1}^n u_k^1u_k^2u_k^3\right)\prod_{i\neq 1,2,3}\left(\frac{1}{n} \displaystyle\sum_{k=1}^n u_k^i\right)+
\left(\frac{1}{n} \displaystyle\sum_{k=1}^n u_k^1u_k^2u_k^4\right)\prod_{i\neq 1,2,4}\left(\frac{1}{n} \displaystyle\sum_{k=1}^n u_k^i\right)+\cdots\right\}
\\ & +(-1)^d\left(\frac{1}{n} \displaystyle\sum_{k=1}^n u_k^1u_k^2\cdots u_k^d\right).
\end{align*}
By the CR-inequality, we get
\begin{align*}
&|\xi_n (t_1, t_2,\dots,t_d)|^2
\\=& C\Bigg[\prod^{d}_{i=1}\left(\frac{1}{n} \displaystyle\sum_{k=1}^n |u_k^i|^2\right)+d^2\prod^{d}_{i=1}\left(\frac{1}{n} \displaystyle\sum_{k=1}^n |u_k^i|^2\right)
\\ & + \left\{
\left(\left|\frac{1}{n} \displaystyle\sum_{k=1}^n u_k^1u_k^2\right|^2\right)\prod_{i\neq 1,2}\left(\frac{1}{n} \displaystyle\sum_{k=1}^n |u_k^i|^2\right)+
\left(\left|\frac{1}{n} \displaystyle\sum_{k=1}^n u_k^1u_k^3\right|^2\right)\prod_{i\neq 1,3}\left(\frac{1}{n} \displaystyle\sum_{k=1}^n u_k^i\right)+\cdots\right\}
\\ & + \left\{
\left(\left|\frac{1}{n} \displaystyle\sum_{k=1}^n u_k^1u_k^2u_k^3\right|^2\right)\prod_{i\neq 1,2,3}\left(\frac{1}{n} \displaystyle\sum_{k=1}^n |u_k^i|^2\right)+
\left(\left|\frac{1}{n} \displaystyle\sum_{k=1}^n u_k^1u_k^2u_k^4\right|^2\right)\prod_{i\neq 1,2,4}\left(\frac{1}{n} \displaystyle\sum_{k=1}^n |u_k^i|^2\right)+\cdots\right\}
\\ & +(-1)^d\left(\left|\frac{1}{n} \displaystyle\sum_{k=1}^n u_k^1u_k^2\cdots u_k^d\right|^2\right)\Bigg],
\end{align*}
for some positive constant $C>0$. By the Cauchy-Schwarz inequality, we have for any $2\leq q\leq d$,
\begin{equation}
\left|\frac{1}{n} \displaystyle\sum_{k=1}^n u_k^1 \, u_k^2\cdots u_k^q\right|^2 \; \leq \; \frac{1}{n} \displaystyle\sum_{k=1}^n \prod_{i \in S_{q_1}} |u_k^i|^2 \; . \; \frac{1}{n} \displaystyle\sum_{k=1}^n \prod_{i \in S_{q_2}} |u_k^i|^2,
\end{equation}
where $S_{q_1}\cup S_{q_2}=\{1,2,\dots,d\}$. By Assumption \ref{ass1} and (\ref{as-1})-(\ref{as-3}), we have
\begin{align*}
&\lim_{\delta \to 0}\, \lim_{n \to \infty} \, \int_{|t_i|_{p_i} < \delta} |\xi_n (t_1, t_2, ..\,, t_d)|^2 \, dw =  0 \quad a.s,\\
&\lim_{\delta \to 0}\, \lim_{n \to \infty} \, \int_{|t_i|_{p_i} > 1/\delta} |\xi_n (t_1, t_2, ..\,, t_d)|^2 \, dw  = 0 \quad a.s,
\end{align*}
for every $i\in\{1,2,\dots,d\}$. This implies that $\displaystyle\lim_{\delta \to 0} \lim_{n \to \infty} D_{n,\delta}^{(2)} \, = \, 0$ almost surely and thus completes the proof.
\end{proof}

\begin{proof}[Proof of Proposition \ref{prop6}]
Define the empirical process
\begin{align*}
\Gamma_n(t)\, = & \,\sqrt{n}\, \xi_n (t_1, t_2, .. , t_d) =
\frac{1}{\sqrt{n}} \displaystyle \sum_{j=1}^n \prod_{i=1}^d
(\hat{f_i} (t_i) - e^{\imath \langle t_i,X_{ji}\rangle}) .
\end{align*}
Then $n\widehat{dcov^2}(X_1, X_2, \dots , X_d)=\Arrowvert
\Gamma_n\Arrowvert^2:=\int \Gamma_n(t_1,t_2,\dots,t_d)^2 dw$. Under
the assumption of independence, we have $\E(\Gamma_n(t))=0$ and
\begin{align*}
\Gamma_n(t) \, \overline{\Gamma_n(t_0)} \,= & \, \frac{1}{n}
\displaystyle \sum_{k,l=1}^n \prod_{i=1}^d (\hat{f_i} (t_i) -
e^{\imath \langle t_i , X_{ki} \rangle})(\hat{f_i} (-t_{i0}) -
e^{-\imath \langle t_{i0} , X_{li} \rangle})
\\ =& \, \frac{1}{n}\,\Bigg\{ \displaystyle \sum_{k=1}^n \prod_{i=1}^d (\hat{f_i} (t_i) - e^{\imath \langle t_i , X_{ki}\rangle})(\hat{f_i} (-t_{i0}) - e^{-\imath \langle t_{i0} , X_{ki} \rangle})
\\ & +\displaystyle \sum_{k\neq l}^n \prod_{i=1}^d (\hat{f_i} (t_i) - e^{\imath \langle t_i , X_{ki} \rangle})(\hat{f_i} (-t_{i0}) - e^{-\imath \langle t_{i0} , X_{li}
\rangle})\Bigg\},
\end{align*}
which implies that
\begin{align*}
\E \big[\Gamma_n(t) \, \overline{\Gamma_n(t_0)}\big] \, =& \,
\frac{1}{n}\,\Big\{ n \, \prod_{i=1}^d \, \E (\hat{f_i} (t_i) -
e^{\imath \langle t_i , X_{ki}\rangle})(\hat{f_i} (-t_{i0}) -
e^{-\imath \langle t_{i0} , X_{ki}\rangle})
\\ & + \,n(n-1) \, \prod_{i=1}^d \E (\hat{f_i} (t_i) - e^{\imath \langle t_i , X_{ki}\rangle})(\hat{f_i} (-t_{i0}) - e^{-\imath \langle t_{i0} , X_{li}\rangle})\,
\Big\}
\\ =&\, \frac{1}{n}\,\big\{ n \, A \; + \; n(n-1) \, B \,\big\}
\quad (\text{say}).
\end{align*}
Direct calculation shows that
\begin{align*}
A \, =& \, \prod_{i=1}^d \, \E \Big\{\frac{1}{n^2} \displaystyle
\sum_{a,b=1}^n e^{\imath \langle t_i , X_{ai}\rangle - \imath
\langle  t_{i0} , X_{bi}\rangle} \, - \,\frac{1}{n} \displaystyle
\sum_{b=1}^n e^{\imath \langle t_i , X_{ki}\rangle-\imath \langle
t_{i0} , X_{bi}\rangle}
\\ & \qquad - \frac{1}{n} \displaystyle \sum_{a=1}^n e^{-\imath \langle t_{i0} , X_{ki}\rangle+\imath \langle t_{i} , X_{ai}\rangle} \, + \, e^{\imath \langle t_i - t_{i0} , X_{ki}\rangle} \,
\Big\}
\\ = & \, \prod_{i=1}^d \,\Big[\, \frac{1}{n^2} \big\{ n\, f_i(t_i - t_{i0}) \, + \, n(n-1) f_i(t_i) f_i(-t_{i0})
\big\}
\\ & \qquad - \, \frac{2}{n} \big\{ f_i(t_i - t_{i0}) \, + \, (n-1) f_i(t_i) f_i(-t_{i0}) \big\} \, + \, f_i(t_i - t_{i0})\,\Big]
\\ = & \, \left(\frac{n-1}{n}\right)^d \,\prod_{i=1}^d \,\big\{f_i(t_i - t_{i0}) \, - \, f_i(t_i) f_i(-t_{i0}) \big\},
\end{align*}
and
\begin{align*}
B \, =& \, \prod_{i=1}^d \, \E \Big[\frac{1}{n^2} \displaystyle
\sum_{a,b=1}^n e^{\imath \langle t_i , X_{ai}\rangle - \imath
\langle t_{i0} , X_{bi}\rangle} \, - \displaystyle \sum_{b=1}^n
e^{\imath \langle t_i , X_{ki}\rangle-\imath \langle t_{i0} ,
X_{bi}\rangle}
\\ & \qquad - \,\frac{1}{n} \displaystyle \sum_{a=1}^n e^{-\imath \langle t_{i0} , X_{li}\rangle+\imath \langle t_{i} , X_{ai}\rangle} \, + \, e^{\imath \langle t_i , X_{ki}\rangle - \imath \langle t_{i0} , X_{li}\rangle} \, \Big]
\\ = & \, \prod_{i=1}^d \,\Big[\, \frac{1}{n^2} \big\{ n\, f_i(t_i - t_{i0}) \, + \, n(n-1) f_i(t_i) f_i(-t_{i0})
\big\}
\\ & \qquad - \,  \frac{2}{n} \big\{ f_i(t_i - t_{i0}) \, + \, (n-1) f_i(t_i) f_i(-t_{i0}) \big\} \, + \, f_i(t_i) f_i(-t_{i0})\, \Big]
\\ = & \, \left(-\frac{1}{n}\right)^d \,\prod_{i=1}^d \,\big\{f_i(t_i - t_{i0}) \, - \, f_i(t_i) f_i(-t_{i0}) \big\}.
\end{align*}
Hence we obtain
\begin{equation}\label{dis1}
\E\big[\Gamma_n(t)\overline{\Gamma_n(t_0)}\big]= c_n\,\prod_{i=1}^d
\,\big\{f_i(t_i - t_{i0}) - f_i(t_i) f_i(-t_{i0}) \big\},
\end{equation}
where
$c_n=\left(\frac{n-1}{n}\right)^d+(n-1)\left(-\frac{1}{n}\right)^d
$. To prove $\Arrowvert \Gamma_n\Arrowvert^2 \,
\overset{d}{\longrightarrow} \, \Arrowvert \Gamma \Arrowvert^2$ , we
construct a sequence of random variables $\{Q_n(\delta)\}$ such that
\begin{enumerate}
\item $Q_n(\delta) \overset{d}{\longrightarrow} Q(\delta)$ as $n \to \infty$, for any fixed $\delta >
0$;
\item $\displaystyle \limsup_{n \to \infty} \, \E |\, Q_n(\delta) - \Arrowvert \Gamma_n \Arrowvert^2 \, | \, \to 0$ as $\delta \to
0$;
\item $Q(\delta)  \overset{d}{\longrightarrow} \Arrowvert \Gamma \Arrowvert^2 $ as $\delta \to 0$.
\end{enumerate}
Then $\Arrowvert \Gamma_n \Arrowvert^2 \,
\overset{d}{\longrightarrow} \, \Arrowvert \Gamma\Arrowvert^2$
follows from Theorem 8.6.2 of Resnick (1999).

We first show (1). Define  \[ Q_n(\delta) \, = \,
\displaystyle\int_{D(\delta)} |\Gamma_n(t)|^2 \, dw,\quad Q(\delta)
\, = \, \displaystyle\int_{D(\delta)} |\Gamma(t)|^2 \, dw.\] Given
$\epsilon > 0$, choose a partition $\{D_k\}_{k=1}^N$ of $D(\delta)$
into $N$ measurable sets with diameter at most $\epsilon$. Then
\[Q_n(\delta) \, = \, \displaystyle \sum_{k=1}^N \int_{D_k}
|\Gamma_n(t)|^2 \, dw,\quad Q(\delta) \, = \, \displaystyle
\sum_{k=1}^N \int_{D_k} |\Gamma(t)|^2 \, dw.\] Define
\[Q_n^{\epsilon}(\delta) \, = \, \displaystyle \sum_{k=1}^N \int_{D_k} |\Gamma_n(t^k)|^2 \, dw,\quad Q^{\epsilon}(\delta) \, = \, \displaystyle \sum_{k=1}^N \int_{D_k} |\Gamma(t^k)|^2 \, dw,\]
where $\{t^k\}_{k=1}^N$ are a set of distinct points such that $t^k
\in D_k$. In view of Theorem $8.6.2$ of Resnick (1999), it suffices
to show that
\begin{enumerate}[ i)]
\item $\displaystyle \limsup_{\epsilon \to 0}\, \limsup_{n \to \infty} \, \E |\, Q_n^{\epsilon}(\delta) - Q_n(\delta)\,| \, = \,
0$;
\item $\displaystyle \limsup_{\epsilon \to 0} \, \E |\, Q^{\epsilon}(\delta) - Q(\delta) \, | \, =\,
0$;
\item $Q_n^{\epsilon}(\delta) \overset{d}{\longrightarrow} Q^{\epsilon}(\delta)$ as $n \to \infty$, for any fixed $\delta > 0$.
\end{enumerate}
To this end, define $\beta_n(\epsilon) \, = \, \sup_{t , t_0} \,\E
\,\big|\, |\Gamma_n(t)|^2 - |\Gamma_n(t_0)|^2\, \big|$ and
$\beta(\epsilon) \, =\, \sup_{t , t_0} \, \E \,\big|\, |\Gamma(t)|^2
- |\Gamma(t_0)|^2\, \big|$, where the supremum is taken over all all
$t=(t_1,\, ..\, , t_d)$ and $t_0=(t_{10},\, ..\, , t_{d0})$ such
that $\delta < |t_i| , |t_{i0}| < \, 1/\delta$ for $i=1,2,\dots,d$,
and $\sum_{i=1}^d |t_i - t_{i0}|_{p_i}^2 \, < \, \epsilon^2$. Since the function inside the supremum is continuous in $t$ and $t_0$, and using the fact that a continuous function on a compact support is uniformly continuous, it follows that $\displaystyle \lim_{\epsilon \to 0} \, \beta(\epsilon) \,
= \, 0$ and $\displaystyle \lim_{\epsilon \to 0} \, \beta_n(\epsilon) \, = \, 0$
for fixed $\delta > 0$ and fixed $n$. 
Thus (i) and (ii) hold. To show (iii), it is enough to show
\[\begin{pmatrix}
  \Gamma_n(t^1)  \\
  \Gamma_n(t^2)  \\
  \vdots    \\
  \Gamma_n(t^N)
 \end{pmatrix} \; \overset{d}{\longrightarrow} \; \begin{pmatrix}
  \Gamma(t^1)  \\
  \Gamma(t^2)  \\
  \vdots    \\
  \Gamma(t^N)
 \end{pmatrix}\, ,\] where $(t^1,\dots, t^N)\in \mathbb{R}^{p_1} \times \mathbb{R}^{p_2} \times \cdots \times \mathbb{R}^{p_d}$ is fixed.
 The rest follows from the Continuous Mapping Theorem and the Cramer-Wold Device.
 Notice that $ \Gamma_n(t)=\frac{1}{\sqrt{n}}\sum_{j=1}^n \prod_{i=1}^d \, \Big[\, \big(\hat{f}_i(t_i) - f_i(t_i)\big) \, - \, \big(e^{\imath \langle t_i , X_{ji}\rangle} -
 f_i(t_i)\big)\,\Big]$. By some algebra and the weak law of large number, we have
 \begin{align*}
 \begin{pmatrix}
  \Gamma_n(t^1)  \\
  \Gamma_n(t^2)  \\
  \vdots    \\
  \Gamma_n(t^N)
 \end{pmatrix}
 =\;\, \frac{1}{\sqrt{n}}\, \displaystyle \sum_{j=1}^n\mathcal{Z}_j\,\; + \;
 o_p(1),
 \end{align*}
 where $\mathcal{Z}_j=(\mathcal{Z}_{j1},\dots,\mathcal{Z}_{jN})'$ with $\mathcal{Z}_{jk}=\prod_{i=1}^d \big(
 f_i(t_i^k)-e^{\imath \langle t_i^k , X_{ji}\rangle}\big)$ for $1\leq k\leq N.$ By the independence assumption, $\E[X_j]=0$ and for $1 \leq l , m \leq N$,
 \begin{align*}
 \E[\mathcal{Z}_{jl}\overline{\mathcal{Z}}_{jm}]\, =& \, \displaystyle \prod_{i=1}^d \, \E \,\big\{e^{\imath \langle t_i^l , X_{ji}\rangle} - f_i(t_i^l)\big\}\big\{e^{-\imath \langle t_i^m , X_{ji}\rangle} - f_i(t_i^m)\big\} = \, R(t^l ,
 t^m).
 \end{align*}
 By the Central Limit Theorem (CLT) and Stutsky's theorem, as $n \to \infty$,
 \[\begin{pmatrix}
  \Gamma_n(t^1)  \\
  \Gamma_n(t^2)  \\
  \vdots    \\
  \Gamma_n(t^N)
 \end{pmatrix} \; \overset{d}{\longrightarrow} \; \begin{pmatrix}
  \Gamma(t^1)  \\
  \Gamma(t^2)  \\
  \vdots    \\
  \Gamma(t^N)
 \end{pmatrix},\]
which completes the proof of (1).

To prove $(2)$, define $u_i \, = \, e^{\imath \langle t_i ,
X_{i}\rangle} - f_i (t_i)$. Then $|u_i|^2 \, = \, 1 + |f_i(t_i)|^2 -
e^{\imath \langle t_i , X_i \rangle} \, \overline{f_i(t_i)} -
e^{-\imath \langle t_i , X_i \rangle}\,f_i(t_i)$, and hence
\begin{equation}\label{dis2}
\E |u_i|^2 \, = \, 1 - |f_i(t_i)|^2  .
\end{equation}
Following the similar steps as in the proof of Theorem 5 in Sz\'{e}kely
\textit{et al.}(2007) and using the Fubini's Theorem,
\begin{equation}\label{th3.2}
\begin{split}
\E \,|\, Q_n(\delta) - \Arrowvert \Gamma_n(t) \Arrowvert^2 \, | \;
=& \; \E \,\big|\int_{D(\delta)} |\Gamma_n(t)|^2 \, dw \, - \, \int
|\Gamma_n(t)|^2 \, dw \,\big|
\\ \leq & \, \int_{|t_1|_{p_1} < \delta} \E|\Gamma_n(t)|^2 \, dw \, + \int_{|t_1|_{p_1} > 1/\delta} \E|\Gamma_n(t)|^2 \, dw \,
\\ & \, + \, \cdots\, +  \int_{|t_d|_{p_d} < \delta} \E\,|\Gamma_n(t)|^2 \, dw \, + \int_{|t_d|_{p_d} > 1/\delta} \E\,|\Gamma_n(t)|^2 \, dw  .
\end{split}
\end{equation}
Using (\ref{dis1}) and (\ref{dis2}), we have $\E\, |\Gamma_n(t)|^2
\,=\,c_n\prod_{i=1}^d \, \E|u_i|^2.$ Along with the independence
assumption, we have
\begin{align*}
\int_{|t_1|_{p_1} < \delta} \E\,|\Gamma_n(t)|^2 \, dw \, =& \,c_n
\int_{|t_1|_{p_1} < \delta} \frac{\E
|u_1|^2}{c_{p_1}\,|t_1|_{p_1}^{1+p_1}} \, dt_1 \prod_{i=2}^d \,
\,\int \frac{\E|u_i|^2}{c_{p_i}\,|t_i|_{p_i}^{1+p_i}} \, dt_i
\\ \leq &\, 2c_n \E |X_1 - X_1^{'}|_{_{p_1}}\,G(|X_1 - X_1^{'}|_{_{p_1}}\delta) \, \prod_{i=2}^d 4\E |X_i|_{_{p_i}}.
\end{align*}
Therefore \[\displaystyle \lim_{\delta \to 0}\, \lim_{n \to
\infty}\, \int_{|t_1|_{p_1} < \delta} \E\,|\Gamma_n(t)|^2 \, dw \, =
\, 0 \; .\] Similarly
\begin{align*}
\int_{|t_1|_{p_1} > 1/\delta} \E\,|\Gamma_n(t)|^2 \, dw \, =& \,c_n
\int_{|t_1|_{p_1} > 1/\delta} \frac{\E
|u_1|^2}{c_{p_1}\,|t_1|_{p_1}^{1+p_1}} \, dt_1 \,. \, \prod_{i=2}^d
\,\int \frac{\E|u_i|^2}{c_{p_i}\,|t_i|_{p_i}^{1+p_i}} \, dt_i
\\ \leq &\, 4\delta c_n\prod_{i=2}^d 4\E |X_i|_{_{p_i}}.
\end{align*}
Therefore \[\displaystyle \lim_{\delta \to 0}\, \lim_{n \to
\infty}\, \int_{|t_1|_{p_1} > 1/\delta} \E\,|\Gamma_n(t)|^2 \, dw \,
= \, 0 \; .\] Applying similar argument to the remaining summands in
($\ref{th3.2}$), we get \[\displaystyle \lim_{\delta \to 0}\, \lim_{n \to
\infty}\, \E |\, Q_n(\delta) - \Arrowvert \Gamma_n(t) \Arrowvert^2
\, | \, = \, 0 \; .\]

To prove $(3)$, we note that \[ \Gamma(t) \,\Huge
\mathbf{1}\normalsize \big(t \in D(\delta)\big) \,\,
\overset{a.s}{\longrightarrow} \,\, \Gamma(t) \,
\Huge\mathbf{1}\normalsize\big(t \in \mathbb{R}^{p_1} \times
\mathbb{R}^{p_2} \times \cdots \times \mathbb{R}^{p_d}\big),\] as
$\delta\rightarrow0.$ Again by the
Fubini's Theorem and equation $(2.5)$ of Sz\'{e}kely \textit{et al.}
(2007),
\begin{align*}
\E \Arrowvert \Gamma \Arrowvert^2 \, =& \, \int \, \prod_{i=1}^d
\big( 1 - |f_i(t_i)|^2\big) \, dw \, = \, \prod_{i=1}^d\, \int
\frac{\big( 1 - |f_i(t_i)|^2 \big)}{c_{p_i}\,|t_i|_{p_i}^{1+p_i}} \,
dt_i
\\ =& \, \prod_{i=1}^d \, \E \int \frac{1 - \cos \langle t_i , X_i-X_i^{'} \rangle}{c_{p_i}\,|t_i|_{p_i}^{1+p_i}} \, dt_i
\\ =& \,\,\prod_{i=1}^d \, \E\, |X_i - X_i^{'}|_{p_i} \, < \, \infty \; .
\end{align*}
Hence $\Arrowvert \Gamma \Arrowvert^2 < \infty$ almost surely. By
DCT, $Q(\delta)  \overset{a.s}{\longrightarrow} \Arrowvert \Gamma
\Arrowvert^2$ as $\delta \to 0$, which completes the proof.
\end{proof}

\begin{lemma}\label{lemma5}
$\widetilde{U}_{i}(k,l)$ can be composed as
\begin{align*}
\widetilde{U}_{i}(k,l)=&\frac{n-3}{(n-1)(n-2)}\sum_{u\notin \{k,l\}}U_i(X_{ui},X_{li})+\frac{n-3}{(n-1)(n-2)}\sum_{v\notin \{k,l\}}U_{i}(X_{ki},X_{vi})
\\&-\frac{n-3}{n-1}U_{i}(X_{ki},X_{li})+\frac{2}{(n-1)(n-2)}\sum_{u,v\notin \{k,l\},u<v}U_i(X_{ui},X_{vi}),
\end{align*}
where the four terms are uncorrelated with each other.
\end{lemma}
\begin{proof}[Proof of Lemma \ref{lemma5}]
The result follows from direct calculation.
\end{proof}

\begin{proposition}\label{prop9}
$\E[\widetilde{dCov}^2(X_i,X_j)]=dCov^2(X_i,X_j)$.
\end{proposition}
\begin{proof}[Proof of Proposition \ref{prop9}]
Using Lemma \ref{lemma5} and the fact that $dCov^2(X_i,X_j)=\E[U_{i}(X_{ki},X_{li})U_{j}(X_{kj},X_{lj})]$ for $k\neq l$, we have for $k\neq l$,
\begin{align*}
&\E[U_{i}(X_{ki},X_{li})U_{j}(X_{kj},X_{lj})]
\\=&\left\{\frac{(n-3)^2}{(n-1)^2}+\frac{2(n-3)^2}{(n-1)^2(n-2)}+\frac{2(n-3)}{(n-1)^2(n-2)}\right\}\E[U_{i}(X_{ki},X_{li})U_{j}(X_{kj},X_{lj})]
\\=&\frac{n-3}{n-1}dCov^2(X_i,X_j).
\end{align*}
It thus implies that
$$\E[\widetilde{dCov^2}(X_i,X_j)]=\frac{n-1}{n-3}\E[U_{i}(X_{ki},X_{li};\alpha)U_{j}(X_{kj},X_{lj})]=dCov^2(X_i,X_j),$$
which completes the proof.
\end{proof}

\begin{proof}[Proof of Proposition \ref{prop10}]
Denote by $\mathbf{X}=\{\mathbf{X}_1,\dots,\mathbf{X}_n\}$. By independence of the bootstrap samples, we have $\E \left[\Gamma^*_n(t)| \,\textbf{X} \right]=0$. Proceeding in the similar way as in the proof of \textsc{Proposition \ref{prop6}}, it can be shown that

\begin{equation}\label{dis11}
\E\big[\Gamma^*_n(t)\overline{\Gamma^*_n(t_0)} \,| \; \textbf{X}\big]= c_n\,\prod_{i=1}^d
\,\big\{\hat{f_i}(t_i - t_{i0}) - \hat{f_i}(t_i) \hat{f_i}(-t_{i0}) \big\},
\end{equation}
where
$c_n=\left(\frac{n-1}{n}\right)^d+(n-1)\left(-\frac{1}{n}\right)^d
$. \\

To prove $\Arrowvert \Gamma^*_n\Arrowvert^2 \,
\overset{d}{\longrightarrow} \, \Arrowvert \Gamma \Arrowvert^2$ almost surely, we
construct a sequence of random variables $\{Q^*_n(\delta)\}$ such that
\begin{enumerate}
\item $Q^*_n(\delta) \overset{d}{\longrightarrow} Q(\delta)$ almost surely as $n \to \infty$, for any fixed $\delta >
0$;
\item $\displaystyle \limsup_{n \to \infty} \, \E \,\left[ \, Q^*_n(\delta) - \Arrowvert \Gamma^*_n \Arrowvert^2 \, | \; \textbf{X} \right] \, \to 0$ almost surely as $\delta \to 0$; \normalsize
\item $Q(\delta)  \overset{d}{\longrightarrow} \Arrowvert \Gamma \Arrowvert^2 $ as $\delta \to 0$.
\end{enumerate}
Then $\Arrowvert \Gamma^*_n \Arrowvert^2 \,
\overset{d}{\longrightarrow} \, \Arrowvert \Gamma\Arrowvert^2$ almost surely follows from Theorem 8.6.2 of Resnick (1999).\\

We first show (1). Define  \[ Q^*_n(\delta) \, = \,
\displaystyle\int_{D(\delta)} |\Gamma^*_n(t)|^2 \, dw,\quad Q(\delta)
\, = \, \displaystyle\int_{D(\delta)} |\Gamma(t)|^2 \, dw.\] Given
$\epsilon > 0$, choose a partition $\{D_k\}_{k=1}^N$ of $D(\delta)$
into $N$ measurable sets with diameter at most $\epsilon$. Then
\[Q^*_n(\delta) \, = \, \displaystyle \sum_{k=1}^N \int_{D_k}
|\Gamma^*_n(t)|^2 \, dw,\quad Q(\delta) \, = \, \displaystyle
\sum_{k=1}^N \int_{D_k} |\Gamma(t)|^2 \, dw.\] Define
\[Q_n^{\epsilon *}(\delta) \, = \, \displaystyle \sum_{k=1}^N \int_{D_k} |\Gamma^*_n(t^k)|^2 \, dw,\quad Q^{\epsilon}(\delta) \, = \, \displaystyle \sum_{k=1}^N \int_{D_k} |\Gamma(t^k)|^2 \, dw,\]
where $\{t^k\}_{k=1}^N$ are a set of distinct points such that $t^k
\in D_k$. In view of Theorem $8.6.2$ of Resnick (1999), it suffices
to show that
\begin{enumerate}[ i)]
\item $\displaystyle \limsup_{\epsilon \to 0}\, \limsup_{n \to \infty} \, \E\,\left[\, |\,Q_n^{\epsilon *}(\delta) - Q^*_n(\delta)\,|\,\big| \, \textbf{X} \right] \, = \,
0$ almost surely ;
\item $\displaystyle \limsup_{\epsilon \to 0} \, \E \left[\, |\, Q^{\epsilon}(\delta) - Q(\delta) \, |\, \right] \, =\,
0$;
\item $Q_n^{\epsilon *}(\delta) \overset{d}{\longrightarrow} Q^{\epsilon}(\delta)$ almost surely as $n \to \infty$, for any fixed $\delta > 0$.
\end{enumerate}
To this end, define \[\beta^*_n(\epsilon) \, = \, \sup_{t , t_0} \,\E
\,\left[ \,\big|\, |\Gamma^*_n(t)|^2 - |\Gamma^*_n(t_0)|^2\,\big| \, \textbf{X}\,\right], \] and,\[ \beta(\epsilon) \, =\,\sup_{t , t_0} \, \E \left[\,\big|\, |\Gamma(t)|^2
- |\Gamma(t_0)|^2\, \big|\,\right],\] where the supremum is taken over all all
$t=(t_1,\, ..\, , t_d)$ and $t_0=(t_{10},\, ..\, , t_{d0})$ such
that $\delta < |t_i| , |t_{i0}| < \, 1/\delta$ for $i=1,2,\dots,d$,
and $\sum_{i=1}^d |t_i - t_{i0}|_{p_i}^2 \, < \, \epsilon^2$. Then
for fixed $\delta > 0$, $\displaystyle \lim_{\epsilon \to 0} \, \beta(\epsilon) \,
= \, 0$ and $\displaystyle \lim_{\epsilon \to 0} \, \beta^*_n(\epsilon) \, = \, 0$ almost surely
for fixed $n$. Thus (i) and (ii)
hold. To show (iii), it is enough to show
\[\begin{pmatrix}
  \Gamma^*_n(t^1)  \\
  \Gamma^*_n(t^2)  \\
  \vdots    \\
  \Gamma^*_n(t^N)
 \end{pmatrix} \; \overset{d}{\longrightarrow} \; \begin{pmatrix}
  \Gamma(t^1)  \\
  \Gamma(t^2)  \\
  \vdots    \\
  \Gamma(t^N)
 \end{pmatrix}\, \; almost \;surely ,\] where $(t^1,\dots, t^N)\in \mathbb{R}^{p_1} \times \mathbb{R}^{p_2} \times \cdots \times \mathbb{R}^{p_d}$ is fixed.
 The rest follows from the Continuous Mapping Theorem and the Cramer-Wold Device.
 Notice that $ \Gamma^*_n(t)=\frac{1}{\sqrt{n}}\sum_{j=1}^n \prod_{i=1}^d \, \Big[\, \big(\hat{f}_i^*(t_i) - \hat{f_i}(t_i)\big) \, - \, \big(e^{\imath \langle t_i , X^*_{ji}\rangle} -
 \hat{f_i}(t_i)\big)\,\Big]$. Using Markov's inequality and Triangle inequality, observe that
 \begin{align*}
 &\displaystyle \sum_{n=1}^\infty P\left(\, \left|\frac{1}{n} \displaystyle \sum_{k=1}^n (e^{\imath \langle t_i , X^*_{ki}\rangle} - \hat{f_i}(t_i) \right| > \epsilon \,\right)
 \\\, =& \, \displaystyle \sum_{n=1}^\infty P \left(\,\left|\displaystyle \sum_{k=1}^n Y_{ki}\right| > n\epsilon \, \right) \; = \; \displaystyle \sum_{n=1}^\infty P \left(\,\left|\displaystyle \sum_{k=1}^n Y_{ki}\right|^2 > n^2 \epsilon^2 \, \right)
\\ =& \, \displaystyle \sum_{n=1}^\infty P \left(\,\displaystyle \sum_{k,l=1}^n Y_{ki} \overline{Y_{li}} > n^2 \epsilon^2 \, \right) \; \leq \;  \displaystyle \sum_{n=1}^\infty \frac{1}{(n\epsilon)^4}\, E \,\left[(\displaystyle \sum_{k,l=1}^n Y_{ki} \overline{Y_{li}})^2 \,\big| \, \textbf{X} \right]
\\ =& \, \displaystyle \sum_{n=1}^\infty \frac{1}{(n\epsilon)^4}\; E \,\left[ \left(\displaystyle \sum_{\substack{k_1,l_1,\\k_2,l_2=1}}^n Y_{k_1 i} \overline{Y_{l_1 i}}\, Y_{k_2 i} \overline{Y_{l_2 i}} \right)\,\big| \, \textbf{X}\,\right]
\\ \leq & \displaystyle \sum_{n=1}^\infty \frac{1}{(n\epsilon)^4} \, . \, C n^2 \;\; < \;\; \infty,
 \end{align*}
where $C>0$, $Y_k = e^{\imath \langle t_i , X^*_{ki}\rangle} - \hat{f_i}(t_i)$, and $|Y_k| \leq 2$ for any $1 \leq k \leq n$.\\

By Borel-Cantelli Lemma, as $n \to \infty$,\, $\hat{f_i^*}(t_i) - \hat{f}_i(t_i) \overset{a.s}{\longrightarrow} 0$ almost surely. By some algebra and the weak law of large number, we have
 \begin{align*}
 \begin{pmatrix}
  \Gamma^*_n(t^1)  \\
  \Gamma^*_n(t^2)  \\
  \vdots    \\
  \Gamma^*_n(t^N)
 \end{pmatrix}
 =\;\, \frac{1}{\sqrt{n}}\, \displaystyle \sum_{j=1}^n\mathcal{Z}_j\,\; + \;
 U \; ,
 \end{align*}
 where $\mathcal{Z}_j=(\mathcal{Z}_{j1},\dots,\mathcal{Z}_{jN})'$ with $\mathcal{Z}_{jk}=\prod_{i=1}^d \big(
 \hat{f_i}(t_i^k)-e^{\imath \langle t_i^k , X^*_{ji}\rangle}\big)$ for $1\leq k\leq N $, and, $U \overset{a.s}{\longrightarrow} 0$, almost surely. By the independence of Bootstrap samples, $\E[\mathcal{Z}_j | \, \textbf{X}]=0$ and for $1 \leq l , m \leq N$,
 \begin{align*}
 \E[\mathcal{Z}_{jl}\overline{\mathcal{Z}}_{jm}]\, =& \, \displaystyle \prod_{i=1}^d \, \E \,\left[ (e^{\imath \langle t_i^l , X^*_{ji}\rangle} - \hat{f_i}(t_i^l))\, (e^{-\imath \langle t_i^m , X^*_{ji}\rangle} - \hat{f_i}(-t_i^m)) | \;\textbf{X} \right]\\
 =& \displaystyle \prod_{i=1}^d \,\big\{\hat{f_i}(t_i^l - t_i^m) - \hat{f_i}(t_i^l) \hat{f_i}(-t_i^m) \big\} .
 \end{align*}
Let $R_n$ and $R$ be $N \times N$ matrices with the $(l, m)^{th}$ element being \[R_n (l, m) \, = \, \displaystyle \prod_{i=1}^d \,\big\{\hat{f_i}(t_i^l - t_i^m) - \hat{f_i}(t_i^l) \hat{f_i}(-t_i^m) \big\} \; , \] and, \[ R (l, m) \, = \, \displaystyle \prod_{i=1}^d \,\big\{f_i(t_i^l - t_i^m) - f_i(t_i^l) f_i(-t_i^m) \big\} \; .\]
By Multivariate CLT,
\begin{align*}
R_n^{-\frac{1}{2}} \, \begin{pmatrix}
  \Gamma^*_n(t^1)  \\
  \Gamma^*_n(t^2)  \\
  \vdots    \\
  \Gamma^*_n(t^N)
 \end{pmatrix} \, \overset{d}{\longrightarrow} \, N(0, I_N) \; \; almost\; surely \; ,
\end{align*}
which, along with the fact $R_n \overset{a.s}{\longrightarrow} R$ and Slutsky's Theorem, implies
 \[\begin{pmatrix}
  \Gamma^*_n(t^1)  \\
  \Gamma^*_n(t^2)  \\
  \vdots    \\
  \Gamma^*_n(t^N)
 \end{pmatrix} \; \overset{d}{\longrightarrow} \; \begin{pmatrix}
  \Gamma(t^1)  \\
  \Gamma(t^2)  \\
  \vdots    \\
  \Gamma(t^N)
 \end{pmatrix} \; \; almost\; surely \;,\]
and thus completes the proof of (1).\\

To prove $(2)$, define $u_i^* \, = \, e^{\imath \langle t_i ,
X^*_{i}\rangle} - \hat{f_i}^* (t_i)$. Then $|u_i|^2 \, = \, 1 + |\hat{f_i}(t_i)|^2 -
e^{\imath \langle t_i , X^*_i \rangle} \, \overline{\hat{f_i}(t_i)} -
e^{-\imath \langle t_i , X^*_i \rangle}\,\hat{f_i}(t_i)$, and hence
\begin{equation}\label{dis22}
\E \,\left[\,|u_i^*|^2 \, | \textbf{X}\,\right] \, = \, 1 - |\hat{f_i}(t_i)|^2  .
\end{equation}

Following the similar steps as in the proof of Theorem 5 in Sz\'{e}kely
\textit{et al.}(2007) and using the Fubini's Theorem,
\begin{equation}\label{th4.1}
\begin{split}
&\E \,\left[\, | Q^*_n(\delta) - \Arrowvert \Gamma^*_n(t)
\Arrowvert^2 |\, \big| \, \textbf{X} \right] \; \\=& \; \E
\,\left[\big|\int_{D(\delta)} |\Gamma^*_n(t)|^2 \, dw \, - \, \int
|\Gamma^*_n(t)|^2 \, dw |\; \big| \, \textbf{X} \right]
\\ \leq & \, \int_{|t_1|_{p_1} < \delta} \E\, \left[\,|\Gamma^*_n(t)|^2\, \big| \, \textbf{X} \right] \, dw \, + \int_{|t_1|_{p_1} > 1/\delta} \E \, \left[\,|\Gamma^*_n(t)|^2\, \big|\, \textbf{X} \right] \, dw \,
\\ & \, + \, \cdots\, +  \int_{|t_d|_{p_d} < \delta} \E\,\left[\,|\Gamma^*_n(t)|^2\, \big| \, \textbf{X} \right] \, dw \, + \int_{|t_d|_{p_d} > 1/\delta} \E\,\left[\,|\Gamma^*_n(t)|^2\, \big| \, \textbf{X} \right] \, dw  .
\end{split}
\end{equation}
Using (\ref{dis11}) and (\ref{dis22}), we have $\E\,\left[ \,|\Gamma^*_n(t)|^2\, \big| \, \textbf{X} \right]
\,=\,c_n\prod_{i=1}^d \, \E\,\left[\,|u^*_i|^2\, \big| \, \textbf{X} \right].$ Along with the independence
assumption, we have
\begin{align*}
&\int_{|t_1|_{p_1} < \delta} \E\,\left[\,|\Gamma^*_n(t)|^2\, \big|
\, \textbf{X} \right] \, dw \, \\=& \;c_n \int_{|t_1|_{p_1} <
\delta} \frac{\E\,\left[\, |u_1^*|^2\, \big| \, \textbf{X}
\right]}{c_{p_1}\,|t_1|_{p_1}^{1+p_1}} \, dt_1 \prod_{i=2}^d \,
\,\int \frac{\E\,\left[\,|u_i^*|^2\, \big| \, \textbf{X}
\right]}{c_{p_i}\,|t_i|_{p_i}^{1+p_i}} \, dt_i
\\ \leq &\; 2 \,c_n\, \E\, \left[\,|X^*_1 - X_1^{*'}|_{_{p_1}}\,G(|X^*_1 - X_1^{*'}|_{_{p_1}}\delta)\,| \, \textbf{X} \right] \, \prod_{i=2}^d 4\,\E\, \left[\,|X^*_i|_{_{p_i}}\, | \, \textbf{X}\right]
\\ =&\; 2 \,c_n\, \frac{1}{n^2}\displaystyle \sum_{j,k=1}^n |X_{j1} - X_{k1}|_{p_1}\,G(|X_{j1} - X_{k1}|_{p_1} \delta)  \, \prod_{i=2}^d 4\,\frac{1}{n}\displaystyle \sum_{j=1}^n |X_{ji}|_{p_i}
\\ \overset{a.s}{\rightarrow} &\; 2 \, \E\, \left[\,|X_1 - X_1^{'}|_{_{p_1}}\,G(|X_1 - X_1^{'}|_{_{p_1}}\delta)\,\right] \, \prod_{i=2}^d 4\,\E\, \left[\,|X_i|_{_{p_i}}\,\right] \qquad as \;\; n \to \infty.
\end{align*}
Therefore \[\displaystyle \lim_{\delta \to 0}\, \lim_{n \to
\infty}\, \int_{|t_1|_{p_1} < \delta} \E\,\left[\,|\Gamma^*_n(t)|^2\, \big| \, \textbf{X} \right] \, dw \, =
\, 0 \;\;\; almost \;surely .\] Similarly
\begin{align*}
\int_{|t_1|_{p_1} > 1/\delta} \E\,\left[\,|\Gamma^*_n(t)|^2\, \big| \, \textbf{X} \right] \, dw \, =& \;c_n
\int_{|t_1|_{p_1} > 1/\delta} \frac{\E\,\left[\,
|u^*_1|^2\, \big| \, \textbf{X} \right]}{c_{p_1}\,|t_1|_{p_1}^{1+p_1}} \, dt_1 \,. \, \prod_{i=2}^d
\,\int \frac{\E\,\left[\,|u^*_i|^2\, \big| \, \textbf{X} \right]}{c_{p_i}\,|t_i|_{p_i}^{1+p_i}} \, dt_i
\\ \leq &\; 4\delta c_n \prod_{i=2}^d 4\,\E \left[\,|X^*_i|_{_{p_i}}|\, \textbf{X}\right].
\end{align*}
Therefore \[\displaystyle \lim_{\delta \to 0}\, \lim_{n \to
\infty}\, \int_{|t_1|_{p_1} > 1/\delta} \E\,\left[\,|\Gamma^*_n(t)|^2\, \big| \, \textbf{X} \right] \, dw \,
= \, 0 \;\;\; almost\; surely .\] Applying similar argument to the remaining summands in
$(\ref{th4.1})$, we get \[\displaystyle \lim_{\delta \to 0}\, \lim_{n \to
\infty}\, \E\,\left[\, |\, Q^*_n(\delta) - \Arrowvert \Gamma^*_n(t) \Arrowvert^2
\, |\, \big| \, \textbf{X} \right] \, = \, 0 \; \; \; almost \; surely .\]
The proof of part $(3)$ is exactly the same as its counterpart in the proof of Proposition \ref{prop6}, which completes the proof.

\end{proof}

Let $G_n$ be the set of all functions from $\{1,2,\dots,n\}$ to $\{1,2,\dots,n\}$. Define a map $g$: $\mathbb{R}^{n\times d}\rightarrow \mathbb{R}^{n\times d}$ as the following
\begin{align*}
g(\X_1,\dots,\X_n)=\begin{pmatrix}
    X_{g_1(1),1} & X_{g_2(1),2} & \dots  & X_{g_d(1),d} \\
    X_{g_1(2),1} & X_{g_2(2),2} & \dots  & X_{g_d(2),d} \\
    \vdots & \vdots & \ddots & \vdots \\
    X_{g_1(n),1} & X_{g_2(n),2} & \dots  & X_{g_d(n),d}
                   \end{pmatrix}
\end{align*}
where $g_i\in G_n$ for $1\leq i\leq d$. With some abuse of notation, we denote by $\widehat{JdCov^2}(g(\X_1,\dots,\X_n))$ the sample (squared) JdCov computed based on the sample $g(\X_1,\dots,\X_n)$. Conditional on the sample, the resampling distribution function $\widehat{F}_n: [0,+\infty) \to [0,1]$ of the bootstrap statistic is defined for all $t \in \mathbb{R}$ as
\begin{align*}
&\widehat{F}_n(\X_1,\dots,\X_n)(t) \,:= \,\frac{1}{n^{nd}}\sum_{g \in G_n^d}\mathds{1}_{\{\, n\widehat{JdCov^2}(g(\X_1,\dots,\X_n))\, \leq \, t\,\}}.
\end{align*}
For $\alpha \in (0,1)$, we define the $\alpha$-level bootstrap-assisted test for testing $H_0$ against $H_A$ as
\begin{equation}\label{one}
\phi_n(\textbf{X}_1, \dots, \textbf{X}_n) := \mathds{1}_{\{\, n\widehat{JdCov^2}\left(\psi(\textbf{X}_1,\, ..\, , \textbf{X}_n) \right) \; > \; \left(\widehat{F}_n(\X_1,\dots,\X_n)\right)^{-1}(1-\alpha)\,\}} \; .
\end{equation}

\begin{proof}[Proof of Proposition \ref{prop11}]
The proof is in similar lines of the proof of Theorem 3.7 in \textit{Pfister et al.} (2018). There exists a set $A_0$ with $P(A_0) = 1$ such that for all $\omega \in A_0$ and $\forall \, t \in \mathbb{R}$,
\begin{align*}
\lim_{n \to \infty} \, \widehat{F}_n (\textbf{X}_1(\omega),\, ..\, , \textbf{X}_n(\omega)) \,(t) \, &= \, \lim_{n \to \infty} \, \frac{1}{n^{nd}} \, \displaystyle \sum_{g \in G_n^d} \mathbbm{1}_{\{\, n\widehat{JdCov^2}\big( g(\textbf{X}_1(\omega),\, ..\, , \textbf{X}_n(\omega))\,\big)\, \leq \, t\,\}}
\\ &= \, \lim_{n \to \infty} \, E\left( \mathbbm{1}_{\{\, n\widehat{JdCov^2}( g (\textbf{X}_1(\omega),\, ..\, , \textbf{X}_n(\omega))\,)\, \leq \, t\, \}} \right)
\\ &= \, \lim_{n \to \infty} \, P \left( n\widehat{JdCov^2}\big( g \big(\textbf{X}_1(\omega),\, ..\, , \textbf{X}_n(\omega))\,\big)\, \leq \, t\, \right)
\\ &= \, G(t) \; ,
\end{align*}
where $G(\cdot)$ is the distribution function of $\displaystyle \sum_{j=1}^{+\infty} \lambda^{'}_j Z_j^2$ . \\

Since $G$ is continuous, for all $\omega \in A_0$ and $\forall \, t \in \mathbb{R}$, we have
\begin{align*}
\lim_{n \to \infty} \, \left( \widehat{F}_n (\textbf{X}_1(\omega),\, ..\, , \textbf{X}_n(\omega))  \,\right)^{-1} (t) \; = \; G^{-1} (t) \, .
\end{align*}
In particular, for all $\omega \in A_0$, we have
\begin{equation}\label{two}
\lim_{n \to \infty} \, \left( \widehat{F}_n (\textbf{X}_1(\omega),\, ..\, , \textbf{X}_n(\omega)) \,\right)^{-1} (1-\alpha) \; = \; G^{-1} (1-\alpha) \, .
\end{equation}
When $H_0$ is true, using Proposition \ref{prop6}, equation $(\ref{one})$ and Corollary $11.2.3$ in Lehmann and Romano (2005), we have
\begin{align*}
&\displaystyle \limsup_{n \to \infty}\, P\left(\, \phi_n(\textbf{X}_1, \dots, \textbf{X}_n) = 1\,\right)
\\  =& \, \displaystyle \limsup_{n \to \infty}\, P\left(\,n\,\widehat{dcov^2} (\textbf{X}_1,\, ..\, , \textbf{X}_n) \; > \; \left(\widehat{F}_n(\textbf{X}_1,\, ..\, , \textbf{X}_n)\,\right)^{-1} (1-\alpha)\,\right)\\
=& \, 1 - \displaystyle \liminf_{n \to \infty}\, P\left(\,n\,\widehat{dcov^2} (\textbf{X}_1,\, ..\, , \textbf{X}_n) \; \leq \; \left(\widehat{F}_n(\textbf{X}_1,\, ..\, , \textbf{X}_n)\,\right)^{-1} (1-\alpha)\,\right)\\
=& 1 - G\left(G^{-1}(1-\alpha)\right) \, = \, 1 - (1-\alpha) \, = \,
\alpha \, .
\end{align*}
This completes the proof of the proposition.
\end{proof}

\begin{proof}[Proof of Proposition \ref{prop12}]
The proof is in similar lines of the proof of Theorem 3.8 in \textit{Pfister et al.} (2018). In the proof of Proposition \ref{prop11}, we showed that there exists a set $A_0$ with $P(A_0) = 1$ such that for all $\omega \in A_0$,
\[\lim_{n \to \infty} \, \left(\widehat{F}_n (\textbf{X}_1(\omega),\, ..\, , \textbf{X}_n(\omega))\right)^{-1} \,(1-\alpha) \; = \; G^{-1} (1-\alpha) \; .\]

Define the set
\begin{equation}
 A_1 \, = \, \{ \, \omega \, : \, \forall \, t\, \in \mathbb{R}, \, \lim_{n \to \infty} \, \mathbbm{1}_{\{\, n\,\widehat{dcov^2}(\textbf{X}_1(\omega),\, ..\, , \textbf{X}_n(\omega))\, \leq t \,\}} \, = \, 0\,\}\; .
\end{equation}
Clearly, $P(A_1) = 1$ and hence $P(A_0 \cap A_1) = 1$. Fix $\omega \in
A_0 \cap A_1$. Then by $(\ref{one})$ and $(\ref{two})$, there exists a constant $t^* \in \mathbb{R}$ such that $\forall n \in \mathbb{N}$,
\begin{align*}
\lim_{n \to \infty} \, \left( \widehat{F}_n (\textbf{X}_1(\omega),\, ..\, , \textbf{X}_n(\omega)) \,\right)^{-1} (1-\alpha) \; &\leq \; t^* \; .
\end{align*}
Therefore,
\begin{align*}
\lim_{n \to \infty} \, \mathbbm{1}_{\{ \,n\,\widehat{dcov^2} (\textbf{X}_1(\omega),\, ..\, , \textbf{X}_n(\omega)) \; \leq \; \left(\widehat{F_n} (\textbf{X}_1(\omega),\, ..\, , \textbf{X}_n(\omega)) \,\right)^{-1} (1-\alpha)  \,\}}
\\ \leq \; \lim_{n \to \infty} \, \mathbbm{1}_{\{\, n\,\widehat{dcov^2} (\textbf{X}_1(\omega),\, ..\, , \textbf{X}_n(\omega)) \, \leq \, t^* \, \}} \;\; = \;\; 0 \; ,
\end{align*}
i.e., $\mathbbm{1}_{\{ \,n\,\widehat{dcov^2} (\textbf{X}_1,\, ..\, , \textbf{X}_n) \; \leq \; \left(\widehat{F_n}(\textbf{X}_1,\, ..\, , \textbf{X}_n)\,\right)^{-1} (1-\alpha)  \,\}} \; \overset{a.s}{\longrightarrow} \; 0$ \, as $n \to \infty$. It follows by dominated convergence theorem that
\begin{align*}
\lim_{n \to \infty} \, P\left(\,n\,\widehat{dcov^2} (\textbf{X}_1,\, ..\, , \textbf{X}_n) \; \leq \; \left(\widehat{F_n}(\textbf{X}_1,\, ..\, , \textbf{X}_n)\,\right)^{-1} (1-\alpha)\,\right) \;
\\= \; \lim_{n \to \infty} \, E\left(\,\mathbbm{1}_{\{ \,n\,\widehat{dcov^2} (\textbf{X}_1,\, ..\, , \textbf{X}_n) \; \leq \; \left(\widehat{F_n}(\textbf{X}_1,\, ..\, , \textbf{X}_n)\,\right)^{-1} (1-\alpha)  \,\}}\,\right) \; = \; 0 \, ,
\end{align*}
which completes the proof of the proposition.
\end{proof}

\begin{table}[!h]\footnotesize
\centering
\caption{Empirical size and power for the bootstrap-assisted joint independence tests (based on the U-statistics) for $c=1$. The results are obtained based on 1000 replications and the number of bootstrap resamples is taken to be 500.}
\label{table:tb}
\begin{tabular}{c ccc cc cc cc cc cc}
\toprule
&&&&\multicolumn{2}{c}{$\widetilde{JdCov^2}$}&\multicolumn{2}{c}{$\widetilde{JdCov^2_S}$}&\multicolumn{2}{c}{$\widetilde{JdCov^2_R}$}&\multicolumn{2}{c}{dHSIC}&\multicolumn{2}{c}{$T_{MT}$}
\\ \cmidrule(r){5-6}\cmidrule(r){7-8}\cmidrule(r){9-10}\cmidrule(r){11-12}\cmidrule(r){13-14}
& & $n$ & $\tilde{d}$ &
10\% & 5\% & 10\% & 5\% & 10\% & 5\% & 10\% & 5\% &  10\% & 5\%  \\
\hline
\multirow{12}{*}{Ex \ref{eg:size}} & (1) & 50 & 5 & 0.097 & 0.049 & 0.110 & 0.059 & 0.099 & 0.045 & 0.102 & 0.047 & 0.099 & 0.042 \\
    & (1) & 50 & 10 & 0.100 & 0.050 & 0.108 & 0.053 & 0.101 & 0.053 & 0.091 & 0.042 & 0.068 & 0.034 \\
    & (2) & 50 & 5 & 0.103 & 0.062 & 0.096 & 0.052 & 0.099 & 0.045 & 0.104 & 0.048 & 0.115 & 0.061 \\
     & (2) & 50 & 10 & 0.119 & 0.062 & 0.121 & 0.056 & 0.101 & 0.053 & 0.105 & 0.041 & 0.106 & 0.056 \\
     & (3) & 50 & 5 & 0.057 & 0.022 & 0.112 & 0.047 & 0.099 & 0.045 & 0.103 & 0.047 & 0.027 & 0.011 \\
     & (3) & 50 & 10 & 0.050 & 0.017 & 0.100 & 0.051 & 0.101 & 0.053 & 0.091 & 0.040 & 0.013 & 0.006 \\
\cline{2-14}
    & (1) & 100 & 5 & 0.101 & 0.05 & 0.098 & 0.057 & 0.091 & 0.042 & 0.088 & 0.038 & 0.098 & 0.052\\
     & (1) & 100 & 10 & 0.105 & 0.045 & 0.085 & 0.043 & 0.102 & 0.053 & 0.091 & 0.038 & 0.098 & 0.059 \\

     & (2)& 100 & 5 & 0.094 & 0.047 & 0.093 & 0.049 & 0.091 & 0.042 & 0.102 & 0.042 & 0.094 & 0.054 \\
     & (2)& 100 & 10 & 0.115 & 0.063 & 0.102 & 0.06 & 0.102 & 0.053 & 0.104 & 0.049 & 0.106 & 0.06  \\

     & (3)&  100 & 5 & 0.08 & 0.034 & 0.115 & 0.058 & 0.091 & 0.042 & 0.095 & 0.038 & 0.043 & 0.019\\
     & (3)& 100 & 10 & 0.066 & 0.025 & 0.104 & 0.052 & 0.102 & 0.053 & 0.111 & 0.047 & 0.021 & 0.005  \\
\hline
\multirow{12}{*}{Ex \ref{eg:power_normal}}& (1)& 50 & 5 & 0.606 & 0.474 & 0.510 & 0.381 & 0.626 & 0.513 & 0.229 & 0.142 & 0.607 & 0.490 \\
    & (1)& 50 & 10 & 0.495 & 0.359 & 0.306 & 0.192 & 0.705 & 0.596 & 0.145 & 0.070 & 0.669 & 0.545 \\
    & (2)& 50 & 5 & 0.813 & 0.720 & 0.732 & 0.632 & 0.835 & 0.751 & 0.342 & 0.219 & 0.805 & 0.706 \\
     & (2)& 50 & 10 & 0.797 & 0.668 & 0.466 & 0.339 & 0.941 & 0.904 & 0.201 & 0.113 & 0.906 & 0.846 \\
     & (3)& 50 & 5 & 0.877 & 0.817 & 0.815 & 0.764 & 0.886 & 0.840 & 0.374 & 0.242 & 0.849 & 0.787 \\
     & (3)& 50 & 10 & 0.848 & 0.749 & 0.521 & 0.396 & 0.960 & 0.917 & 0.174 & 0.096 & 0.942 & 0.897 \\
\cline{2-14}
	& (1)& 100 & 5 & 0.903 & 0.854 & 0.834 & 0.767 & 0.93 & 0.881 & 0.405 & 0.278 & 0.913 & 0.863 \\
     & (1)& 100 & 10 & 0.853 & 0.756 & 0.468 & 0.337 & 0.977 & 0.954 & 0.203 & 0.114 & 0.97 & 0.936\\

     & (2)& 100 & 5 & 0.989 & 0.981 & 0.968 & 0.946 & 0.99 & 0.983 & 0.618 & 0.491 & 0.987 & 0.975 \\
     & (2)& 100 & 10 & 0.998 & 0.988 & 0.79 & 0.657  &  1  &  1 & 0.36 & 0.215 & 1 & 0.999\\

     & (3)& 100 & 5 & 0.998 & 0.994 & 0.988 & 0.98 & 0.997 & 0.991 & 0.649 & 0.518 & 0.995 & 0.991\\
     & (3)& 100 & 10 & 0.998 & 0.991 & 0.816 & 0.721  &  1  &  1 & 0.307 & 0.189  &  1 & 0.999\\
\hline
\multirow{4}{*}{Ex \ref{eg:mutua1-2}}& (1) & 50 & 3 & 0.998 & 0.986 & 1.000 & 1.000 & 0.624 & 0.365 & 0.898 & 0.794 & 0.221 & 0.106 \\
   & (2) & 50 & 3 &   1 &  1 &  1  & 1 &  1 &  1 &  1  & 1 &  1 &  1 \\
   & (1) & 100 & 3 & 1  &  1  &  1  &  1 &   1& 0.999  &  1  &  1 & 0.622 & 0.368 \\
   & (2) & 100 & 3 & 1  &  1  &  1  &  1 &   1  &  1  &  1  &  1  &  1  &   1\\
\hline
\multirow{8}{*}{Ex \ref{eg:mutua1-3}}

& (1)& 100 & 5 & 0.339 & 0.195 & 0.523 & 0.379 & 0.122 & 0.07 & 0.219 & 0.114 & 0.073 & 0.038 \\
   & (1) & 100 & 10 & 0.105 & 0.027 & 0.248 & 0.147 & 0.049 & 0.019 & 0.117 & 0.043 & 0.025 & 0.008 \\
   & (2) & 100 & 5 & 0.369 & 0.235 & 0.466 & 0.362 & 0.162 & 0.09 & 0.406 & 0.25 & 0.241 & 0.161    \\
   & (2) & 100 & 10 & 0.097 & 0.04 & 0.218 & 0.13 & 0.06 & 0.021 & 0.164 & 0.077 & 0.046 & 0.022 \\

\cline{2-14}

& (1)& 200 & 5 & 0.813 & 0.676 & 0.929 & 0.865 & 0.238 & 0.128 & 0.378 & 0.224 & 0.085 & 0.044 \\
   & (1) & 200 & 10 & 0.262 & 0.140 & 0.433 & 0.305 & 0.093 & 0.045 & 0.137 & 0.061 & 0.047 & 0.023  \\
   & (2) & 200 & 5 &  0.773 & 0.662 & 0.778 & 0.689 & 0.398 & 0.263 & 0.797 & 0.665 & 0.581 & 0.505   \\
   & (2) & 200 & 10 & 0.290 & 0.171 & 0.384 & 0.296 & 0.136 & 0.065 & 0.300 & 0.173 & 0.141 & 0.077  \\

\toprule
\end{tabular}
\\ Note: In Examples \ref{eg:size}-\ref{eg:mutua1-2}, $\tilde{d}$ denotes the number of random variables $d$. In Example \ref{eg:mutua1-3}, $\tilde{d}$ stands for $p$.
\end{table}



\begin{table}[!h]\footnotesize
\centering
\caption{Empirical size and power for the bootstrap-assisted joint independence tests (based on the U-statistics) for $c=2$ and $0.5$. The results are obtained based on 1000 replications and the number of bootstrap resamples is taken to be 500.}
\label{table:tb1}
\begin{tabular}{c ccc cc cc cc | cc cc cc}
\toprule
&&&& \multicolumn{6}{c}{$c=2$} & \multicolumn{6}{c}{$c=0.5$}\\
&&&&\multicolumn{2}{c}{$\widetilde{JdCov^2}$}&\multicolumn{2}{c}{$\widetilde{JdCov^2_S}$}&\multicolumn{2}{c}{$\widetilde{JdCov^2_R}$}&\multicolumn{2}{c}{$\widetilde{JdCov^2}$}&\multicolumn{2}{c}{$\widetilde{JdCov^2_S}$}&\multicolumn{2}{c}{$\widetilde{JdCov^2_R}$}
\\ \cmidrule(r){5-6}\cmidrule(r){7-8}\cmidrule(r){9-10}\cmidrule(r){11-12}\cmidrule(r){13-14}\cmidrule(r){15-16}
& & $n$ & $\tilde{d}$ &
10\% & 5\% & 10\% & 5\% & 10\% & 5\% & 10\% & 5\% & 10\% & 5\% & 10\% & 5\% \\
\hline
\multirow{12}{*}{Ex \ref{eg:size}} & (1) & 50 & 5 & 0.097 & 0.05 & 0.099 & 0.052 & 0.094 & 0.045 & 0.102 & 0.05 & 0.115 & 0.061 & 0.099 & 0.054 \\
    & (1) & 50 & 10 & 0.103 & 0.049 & 0.112 & 0.056 & 0.097 & 0.048 & 0.102 & 0.051 & 0.116 & 0.067 & 0.107 & 0.051\\
    & (2) & 50 & 5 & 0.106 & 0.057 & 0.102 & 0.058 & 0.094 & 0.045 & 0.113 & 0.053 & 0.110 & 0.058 & 0.099 & 0.054 \\
     & (2) & 50 & 10 & 0.107 & 0.051 & 0.107 & 0.06 & 0.097 & 0.048 & 0.125 & 0.074 & 0.120 & 0.071 & 0.107 & 0.051 \\
     & (3) & 50 & 5 & 0.063 & 0.017 & 0.101 & 0.048 & 0.094 & 0.045 & 0.058 & 0.019 & 0.105 & 0.052 & 0.099 & 0.054 \\
     & (3) & 50 & 10 & 0.056 & 0.022 & 0.100 & 0.053 & 0.097 & 0.048 & 0.026 & 0.009 & 0.096 & 0.049 & 0.107 & 0.051 \\
\cline{2-16}
    & (1) & 100 & 5 & 0.087 & 0.043 & 0.098 & 0.049 & 0.085 & 0.046 & 0.097 & 0.059 & 0.107 & 0.066 & 0.098 & 0.042 \\
     & (1) & 100 & 10 & 0.104 & 0.049 & 0.107 & 0.050 & 0.098 & 0.052 & 0.087 & 0.040 & 0.117 & 0.056 & 0.104 & 0.053\\

     & (2)& 100 & 5 & 0.088 & 0.046 & 0.091 & 0.039 & 0.085 & 0.046 & 0.104 & 0.059 & 0.108 & 0.057 & 0.098 & 0.042\\
     & (2)& 100 & 10 & 0.099 & 0.060 & 0.105 & 0.065 & 0.098 & 0.052 & 0.101 & 0.060 & 0.101 & 0.054 & 0.104 & 0.053\\

     & (3)&  100 & 5 & 0.080 & 0.034 & 0.113 & 0.057 & 0.085 & 0.046 & 0.086 & 0.034 & 0.120 & 0.063 & 0.098 & 0.042 \\
     & (3)& 100 & 10 & 0.077 & 0.029 & 0.117 & 0.053 & 0.098 & 0.052 & 0.044 & 0.019 & 0.100 & 0.055 & 0.104 & 0.053 \\
\hline
\multirow{12}{*}{Ex \ref{eg:power_normal}}& (1)& 50 & 5 & 0.644 & 0.526 & 0.629 & 0.504 & 0.630 & 0.517 & 0.434 & 0.323 & 0.291 & 0.196 & 0.610 & 0.499 \\
    & (1)& 50 & 10 & 0.690 & 0.580 & 0.603 & 0.473 & 0.718 & 0.610 & 0.220 & 0.125 & 0.163 & 0.105 & 0.615 & 0.498 \\
    & (2)& 50 & 5 & 0.857 & 0.777 & 0.836 & 0.750 & 0.837 & 0.760 & 0.641 & 0.519 & 0.439 & 0.318 & 0.816 & 0.734 \\
     & (2)& 50 & 10 & 0.944 & 0.887 & 0.872 & 0.798 & 0.953 & 0.914 & 0.313 & 0.212 & 0.221 & 0.165 & 0.887 & 0.811 \\
     & (3)& 50 & 5 & 0.903 & 0.851 & 0.889 & 0.835 & 0.892 & 0.846 & 0.773 & 0.692 & 0.596 & 0.510 & 0.876 & 0.821 \\
     & (3)& 50 & 10 & 0.957 & 0.918 & 0.912 & 0.842 & 0.967 & 0.929 & 0.370 & 0.254 & 0.266 & 0.198 & 0.915 & 0.868 \\
\cline{2-16}
	& (1)& 100 & 5 & 0.935 & 0.890 & 0.912 & 0.877 & 0.932 & 0.886 & 0.747 & 0.637 & 0.453 & 0.346 & 0.916 & 0.867\\
     & (1)& 100 & 10 & 0.979 & 0.943 & 0.927 & 0.860 & 0.983 & 0.963 & 0.308 & 0.194 & 0.188 & 0.129 & 0.949 & 0.890 \\

     & (2)& 100 & 5 & 0.994 & 0.987 & 0.991 & 0.986 & 0.991 & 0.983 & 0.938 & 0.897 & 0.705 & 0.605 & 0.988 & 0.981\\
     & (2)& 100 & 10 & 1 & 1 & 1 & 0.999 & 1 & 1 & 0.476 & 0.352 & 0.274 & 0.210 & 1 & 1 \\

     & (3)& 100 & 5 & 0.998 & 0.997 & 0.998 & 0.994 & 0.997 & 0.991 & 0.980 & 0.962 & 0.872 & 0.817 & 0.997 & 0.991\\
     & (3)& 100 & 10 & 1 & 1 & 1 & 0.999 & 1 & 1 & 0.559 & 0.444 & 0.336 & 0.274 & 1 & 0.998\\
\hline
\multirow{4}{*}{Ex \ref{eg:mutua1-2}}& (1) & 50 & 3 & 0.797 & 0.567 & 0.978 & 0.893 & 0.267 & 0.155& 1 & 1 & 1 & 1 & 1 & 0.984 \\
   & (2) & 50 & 3 & 1 & 1 & 1 & 1 & 0.959 & 0.593 & 1 & 1 & 1 & 1 & 1 & 1 \\
   & (1) & 100 & 3 & 1 & 0.999 & 1 & 1 & 0.704 & 0.458 & 1 & 1 & 1 & 1 & 1 & 1 \\
   & (2) & 100 & 3 & 1 & 1 & 1 & 1 & 1 & 1 & 1 & 1 & 1 & 1 & 1 & 1\\
\hline
\multirow{8}{*}{Ex \ref{eg:mutua1-3}}

& (1)& 100 & 5 & 0.198 & 0.087 & 0.295 & 0.195 & 0.109 & 0.047 & 0.605 & 0.413 & 0.768 & 0.638 & 0.178 & 0.096 \\
   & (1) & 100 & 10 & 0.074 & 0.018 & 0.171 & 0.092 & 0.045 & 0.017 & 0.149 & 0.046 & 0.357 & 0.221 & 0.050 & 0.020 \\
   & (2) & 100 & 5 & 0.342 & 0.221 & 0.444 & 0.315 & 0.180 & 0.095 & 0.438 & 0.338 & 0.496 & 0.419 & 0.267 & 0.143 \\
   & (2) & 100 & 10 & 0.083 & 0.034 & 0.179 & 0.105 & 0.055 & 0.016 & 0.134 & 0.056 & 0.266 & 0.176 & 0.066 & 0.027 \\
\cline{2-16}

& (1)& 200 & 5 & 0.435 & 0.293 & 0.619 & 0.462 & 0.162 & 0.083 & 0.981 & 0.951 & 0.995 & 0.987 & 0.438 & 0.281\\
   & (1) & 200 & 10 & 0.146 & 0.063 & 0.243 & 0.146 & 0.077 & 0.032 & 0.465 & 0.308 & 0.664 & 0.528 & 0.132 & 0.057 \\
   & (2) & 200 & 5 & 0.698 & 0.571 & 0.781 & 0.669 & 0.338 & 0.212 & 0.715 & 0.623 & 0.688 & 0.611 & 0.534 & 0.400  \\
   & (2) & 200 & 10 & 0.214 & 0.129 & 0.316 & 0.213 & 0.120 & 0.052 & 0.349 & 0.241 & 0.442 & 0.352 & 0.169 & 0.082\\

\toprule
\end{tabular}
\\Note: In Examples \ref{eg:size}-\ref{eg:mutua1-2}, $\tilde{d}$ denotes the number of random variables $d$. In Example \ref{eg:mutua1-3}, $\tilde{d}$ stands for $p$.
\end{table}

\begin{table}[!h]\footnotesize
\centering
\caption{Empirical size and power for the bootstrap-assisted joint independence tests based on the V-statistic type estimators, with $c=1$. The results are obtained based on 1000 replications and the number of bootstrap resamples is taken to be 500.}
\label{table:tb_bias}

\begin{tabular}{c ccc  cc cc cc}
\toprule
&&&&\multicolumn{2}{c}{$\widehat{JdCov^2}$}&\multicolumn{2}{c}{$\widehat{JdCov^2_S}$}&\multicolumn{2}{c}{$\widehat{JdCov^2_R}$}
\\ \cmidrule(r){5-6}\cmidrule(r){7-8}\cmidrule(r){9-10}
& & $n$ & $\tilde{d}$ &
10\% & 5\% & 10\% & 5\% & 10\% & 5\%  \\
\hline
\multirow{12}{*}{Ex \ref{eg:size}} & (1) & 50 & 5 & 0.093 & 0.033 & 0.269 & 0.131 & 0.103 & 0.052   \\
    & (1) & 50 & 10 & 0.130 & 0.067 & 0.257 & 0.139 & 0.110 & 0.063 \\
    & (2) & 50 & 5 &  0.142 & 0.081 & 0.106 & 0.061 & 0.103 & 0.052  \\
     & (2) & 50 & 10 & 0.452 & 0.130 & 0.077 & 0.020 & 0.110 & 0.063  \\
     & (3) & 50 & 5 & 0.118 & 0.067 & 0.200 & 0.118 & 0.103 & 0.052  \\
     & (3) & 50 & 10 & 0.124 & 0.069 & 0.195 & 0.111 & 0.110 & 0.063  \\
\cline{2-10}
    & (1) & 100 & 5 & 0.068 & 0.024 & 0.204 & 0.113 & 0.090 & 0.044 \\
     & (1) & 100 & 10 & 0.086 & 0.042 & 0.184 & 0.092 & 0.107 & 0.058 \\

     & (2)& 100 & 5 & 0.121 & 0.061 & 0.102 & 0.053 & 0.090 & 0.044\\
     & (2)& 100 & 10 & 0.222 & 0.050 & 0.056 & 0.013 & 0.107 & 0.058 \\

     & (3)&  100 & 5 & 0.128 & 0.066 & 0.191 & 0.116 & 0.090 & 0.044  \\
     & (3)& 100 & 10 & 0.114 & 0.061 & 0.168 & 0.102 & 0.107 & 0.058 \\
\hline
\multirow{12}{*}{Ex \ref{eg:power_normal}}& (1)& 50 & 5 & 0.485 & 0.299 & 0.649 & 0.450 & 0.637 & 0.528\\
    & (1)& 50 & 10 & 0.284 & 0.161 & 0.428 & 0.271 & 0.727 & 0.627 \\
    & (2)& 50 & 5 & 0.746 & 0.571 & 0.806 & 0.659 & 0.846 & 0.768\\
     & (2)& 50 & 10 & 0.393 & 0.250 & 0.544 & 0.371 & 0.955 & 0.911 \\
     & (3)& 50 & 5 & 0.822 & 0.725 & 0.877 & 0.788 & 0.895 & 0.848  \\
     & (3)& 50 & 10 & 0.479 & 0.325 & 0.637 & 0.459 & 0.965 & 0.938  \\
\cline{2-10}
	& (1)& 100 & 5 & 0.850 & 0.717 & 0.830 & 0.693 & 0.932 & 0.886 \\
     & (1)& 100 & 10 & 0.298 & 0.168 & 0.428 & 0.276 & 0.980 & 0.955\\

     & (2)& 100 & 5 & 0.985 & 0.947 & 0.974 & 0.922 & 0.992 & 0.985  \\
     & (2)& 100 & 10 & 0.500 & 0.328 & 0.595 & 0.436 & 1.000 & 1.000  \\

     & (3)& 100 & 5 & 0.995 & 0.983 & 0.989 & 0.977 & 0.998 & 0.993\\
     & (3)& 100& 10 & 0.613 & 0.441 & 0.700 & 0.551 & 1.000 & 1.000 \\
\hline
\multirow{4}{*}{Ex \ref{eg:mutua1-2}}& (1) & 50 & 3 & 0.985 & 0.928 & 0.999 & 0.997 & 0.647 & 0.377 \\
   & (2) & 50 & 3 & 1 & 1 & 1 & 1 & 1 & 1  \\
   & (1) & 100 & 3 & 1 & 1 & 1 & 1 & 1 & 0.999 \\
   & (2) & 100 & 3 & 1 & 1 & 1 & 1 & 1 & 1   \\

\toprule

\end{tabular}
\\Note: $\tilde{d}$ denotes the number of random variables $d$.
\end{table}

\begin{table}[!h]\footnotesize
\centering
\caption{Empirical size and power for the bootstrap-assisted joint independence tests (based on the U-statistics) with $c$ chosen according to the heuristic idea discussed in Remark \ref{choice_c}. The results are obtained based on 1000 replications and the number of bootstrap resamples is taken to be 500.}
\label{table:tb3}

\begin{tabular}{c ccc cc cc cc cc}
\toprule
&&&&\multicolumn{2}{c}{$c$}&\multicolumn{2}{c}{$\widetilde{JdCov^2}$}&\multicolumn{2}{c}{$\widetilde{JdCov^2_S}$}&\multicolumn{2}{c}{$\widetilde{JdCov^2_R}$}
\\ \cmidrule(r){5-6}\cmidrule(r){7-8}\cmidrule(r){9-10}\cmidrule(r){11-12}
& & $n$ & $\tilde{d}$ &
10\% & 5\% & 10\% & 5\% & 10\% & 5\% & 10\% & 5\%  \\
\hline
\multirow{12}{*}{Ex \ref{eg:size}} & (1) & 50 & 5 & 1.646 & 1.724 & 0.103 & 0.053 & 0.107 & 0.055 & 0.107 & 0.053 \\
    & (1) & 50 & 10 & 1.657 & 1.732 & 0.101 & 0.049 & 0.106 & 0.056 & 0.095 & 0.052 \\
    & (2) & 50 & 5 & 0.440 & 0.533 & 0.116 & 0.055 & 0.114 & 0.058 & 0.110 & 0.052  \\
     & (2) & 50 & 10 & 1.519 & 1.636 & 0.099 & 0.052 & 0.087 & 0.045 & 0.094 & 0.050 \\
     & (3) & 50 & 5 & 0.438 & 0.527 & 0.050 & 0.020 & 0.113 & 0.048 & 0.110 & 0.052 \\
     & (3) & 50 & 10 & 0.438 & 0.527 & 0.027 & 0.011 & 0.094 & 0.051 & 0.107 & 0.048 \\
\cline{2-12}
    & (1) & 100 & 5 & 1.657 & 1.731 & 0.102 & 0.047 & 0.105 & 0.054 & 0.089 & 0.046 \\
     & (1) & 100 & 10 & 1.656 & 1.731 & 0.108 & 0.049 & 0.101 & 0.046 & 0.101 & 0.060 \\

     & (2)& 100 & 5 & 0.438 & 0.527 & 0.112 & 0.063 & 0.109 & 0.058 & 0.098 & 0.044\\
     & (2)& 100 & 10 & 0.484 & 0.620 & 0.104 & 0.064 & 0.104 & 0.048 & 0.116 & 0.066\\

     & (3)&  100 & 5 & 0.438 & 0.527 & 0.082 & 0.039 & 0.116 & 0.070 & 0.098 & 0.044 \\
     & (3)& 100 & 10 & 0.438 & 0.527 & 0.051 & 0.020 & 0.100 & 0.051 & 0.107 & 0.058 \\
\hline
\multirow{12}{*}{Ex \ref{eg:power_normal}}& (1)& 50 & 5 & 1.646 & 1.724 & 0.637 & 0.517 & 0.603 & 0.484 & 0.630 & 0.502 \\
    & (1)& 50 & 10 & 1.657 & 1.732 & 0.651 & 0.517 & 0.529 & 0.403 & 0.718 & 0.600 \\
    & (2)& 50 & 5 & 1.646 & 1.724 & 0.842 & 0.761 & 0.815 & 0.728 & 0.844 & 0.760 \\
     & (2)& 50 & 10 & 1.657 & 1.732 & 0.906 & 0.834 & 0.801 & 0.706 & 0.948 & 0.909 \\
     & (3)& 50 & 5 & 1.646 & 1.724 & 0.901 & 0.844 & 0.882 & 0.819 & 0.889 & 0.845  \\
     & (3)& 50 & 10 & 1.657 & 1.732 & 0.928 & 0.871 & 0.843 & 0.766 & 0.957 & 0.919  \\
\cline{2-12}
	& (1)& 100 & 5 & 1.657 & 1.731 & 0.923 & 0.884 & 0.891 & 0.845 & 0.929 & 0.883  \\
     & (1)& 100 & 10 & 1.656 & 1.731 & 0.951 & 0.905 & 0.867 & 0.778 & 0.982 & 0.953 \\

     & (2)& 100 & 5 & 1.657 & 1.731 & 0.990 & 0.986 & 0.985 & 0.977 & 0.992 & 0.985 \\
     & (2)& 100 & 10 & 1.656 & 1.731 & 0.986 & 0.982 & 0.976 & 0.962 & 1.000 & 1.000 \\

     & (3)& 100 & 5 & 1.657 & 1.731 & 0.998 & 0.996 & 0.996 & 0.990 & 0.996 & 0.992  \\
     & (3)& 100& 10 & 1.656 & 1.731 & 0.991 & 0.984 & 0.974 & 0.965 & 1.000 & 0.999\\
\hline
\multirow{4}{*}{Ex \ref{eg:mutua1-2}}& (1) & 50 & 3 & 0.554 & 0.729 & 0.984 & 0.962 & 0.998 & 0.994 & 0.899 & 0.843  \\
   & (2) & 50 & 3 & 0.438 & 0.527 & 1.000 & 1.000 & 1.000 & 1.000 & 1.000 & 1.000\\
   & (1) & 100 & 3 & .439 & 0.530 & 1.000 & 1.000 & 1.000 & 1.000 & 1.000 & 1.000\\
   & (2) & 100 & 3 & 0.438 & 0.527 & 1.000 & 1.000 & 1.000 & 1.000 & 1.000 & 1.000 \\

\hline
\multirow{8}{*}{Ex \ref{eg:mutua1-3}}

& (1)& 100 & 5 & 1.427 & 1.545 & 0.300 & 0.166 & 0.417 & 0.314 & 0.131 & 0.064 \\
   & (1) & 100 & 10 & 1.589 & 1.680 & 0.090 & 0.023 & 0.204 & 0.111 & 0.050 & 0.018  \\
   & (2) & 100 & 5 & 0.537 & 0.659 & 0.398 & 0.293 & 0.468 & 0.375 & 0.233 & 0.129 \\
   & (2) & 100 & 10 & 1.040 & 1.198 & 0.106 & 0.040 & 0.212 & 0.132 & 0.058 & 0.018  \\
\cline{2-12}

& (1)& 200 & 5 & 1.177 & 1.350 & 0.681 & 0.568 & 0.804 & 0.720 & 0.250 & 0.153 \\
   & (1) & 200 & 10 & 1.503 & 1.609 & 0.221 & 0.117 & 0.340 & 0.234 & 0.086 & 0.040 \\
   & (2) & 200 & 5 & 0.440 & 0.532 & 0.709 & 0.621 & 0.681 & 0.603 & 0.539 & 0.392  \\
   & (2) & 200 & 10 & 0.606 & 0.735 & 0.290 & 0.179 & 0.381 & 0.277 & 0.149 & 0.065 \\

\toprule

\end{tabular}
\\Note: In Examples \ref{eg:size}-\ref{eg:mutua1-2}, $\tilde{d}$ denotes the number of random variables $d$. In Example \ref{eg:mutua1-3}, $\tilde{d}$ stands for $p$.
\end{table}

\end{document}